\documentclass[runningheads]{llncs}

\usepackage[hidelinks]{hyperref}
\usepackage{amsmath, amssymb, amsfonts}
\usepackage[capitalise,noabbrev]{cleveref}
\usepackage{mdframed}
\usepackage{algorithmicx}
\usepackage{algorithm}
\usepackage[noend]{algpseudocode}
\usepackage{lipsum}
\usepackage[title]{appendix}
\usepackage{graphicx}
\usepackage{xcolor}
\usepackage{tikz}
\usetikzlibrary{automata, positioning, arrows, arrows.meta}
\usepackage{float}
\usepackage{xspace}
\usepackage{xfrac}
\usepackage{booktabs}
\usepackage{rotating}
\usepackage{pgfplots}
\usepackage{xifthen}
\usepackage{paralist}
\usepackage{subfigure}
\usepackage{thmtools}
\usepackage{thm-restate}
\usepackage{mathtools}
\usepackage{multirow}
\usepackage{array}
\usepackage{makecell}
\usepackage{microtype}
\usepackage{adjustbox}
\usepackage{caption}
\algnewcommand{\IfThenElse}[3]{%
  \State \algorithmicif\ #1\ \algorithmicthen\ #2\ \algorithmicelse\ #3}

\renewcommand{\paragraph}[1]{\smallskip\noindent\emph{#1}}
\renewcommand{\subsubsection}[1]{\medskip\noindent\textbf{#1}}

\title{Robust Almost-Sure Reachability in Multi-Environment MDPs}
\ifdefined\DOUBLEBLIND
\author{}
\authorrunning{}
\institute{}
\else
\author{Marck van der Vegt\orcidID{0000-0003-2451-5466} \and Nils Jansen\orcidID{0000-0003-1318-8973} \and Sebastian Junges\orcidID{0000-0003-0978-8466}}
\authorrunning{M.\ van der Vegt  \and N.\ Jansen  \and S.\ Junges}
\institute{
  Radboud University, Nijmegen, the Netherlands\\%
  \email{\{marck.vandervegt,  nils.jansen, sebastian.junges\}@ru.nl}%
}
\fi

\makeatletter
\def\@citecolor{blue}%
\def\@urlcolor{blue}%
\def\@linkcolor{RedViolet}%

\def\orcidID#1{\smash{\href{http://orcid.org/#1}{\protect\raisebox{-1.25pt}{\protect\includegraphics{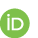}}}}}
\makeatother

\newcommand{\storm}{\textsc{Storm}\xspace}
\newcommand{\prism}{\textsc{Prism}\xspace}
\newcommand{\PaGE}{\textsc{PaGE}}

\newcommand{\unif}{\mathsf{unif}}
\newcommand{\dirac}{\mathsf{dirac}}
\newcommand{\assign}{\leftarrow}
\newcommand{\tuple}[1]{\langle #1 \rangle}
\newcommand{\mdp}{\mathcal{M}}
\newcommand{\memdp}{\mathcal{N}}
\newcommand{\game}{\mathcal{B}}
\newcommand{\egraph}{\mathit{GE}_\memdp}
\newcommand{\egraphV}{\mathit{V}_\memdp}
\newcommand{\egraphE}{\mathit{E}_\memdp}
\newcommand{\custombsg}[1]{\game_{#1}}
\newcommand{\bsg}{\custombsg{\memdp}}
\newcommand{\supp}{\mathit{Supp}}
\newcommand{\dist}{\mathit{Dist}}
\newcommand{\last}{\mathit{last}}

\newcommand{\powerset}[1]{\mathcal{P}\left(#1\right)}

\newcommand{\Path}{\textsc{Path}}

\newcommand{\restrictenv}[2]{{#1}_{\downarrow#2}}

\newcommand{\reachable}[1]{\mathsf{Reachable}(#1)}
\newcommand{\onlyreachable}[2]{\mathsf{ReachFragment}(#1, #2)}

\newcommand{\cut}[2]{{#1}{\mid\!#2}}

\newcommand{\envact}{\alpha_\otimes}

\newcommand{\wintmp}{W}
\newcommand{\notwin}{L}

\newcommand{\pa}{1} %
\newcommand{\pb}{2}

\newcommand{\initstate}{s_\iota}
\newcommand{\exampleend}{\hfill$\blacksquare$}

\newcommand{\win}[2]{\mathsf{Win}_{#1}^{#2}}

\newcommand{\scatterplotsize}[0]{0.32\textwidth}
\newcommand{\scatterplot}[6]{%
	\begin{tikzpicture}
	\begin{axis}[
	width=\scatterplotsize,
	height=\scatterplotsize,
	axis equal image,
	xmin=0.04,
	ymin=0.04,
	ymax=9000,
	xmax=9000,
	xmode=log,
	ymode=log,
	axis x line=bottom,
	axis y line=left,
	xtick={1,9,90,900},
	xticklabels={1,9,90,900},
	extra x ticks = {2700, 9000},
	extra x tick labels = {TO, MO},
	extra x tick style = {grid = major},
	ytick={1,9,90,900},
	yticklabels={1,9,90,900},
	extra y ticks = {2700, 9000},
	extra y tick labels = {TO, MO},
	extra y tick style = {grid = major},
	xlabel={\scriptsize #4},
	xlabel style={yshift=0cm},
	ylabel={\scriptsize #6},
	ylabel style={yshift=-0.4cm},
	yticklabel style={font=\tiny},
	xticklabel style={rotate=290,anchor=west,font=\tiny},
	]

	\addplot[
	scatter,
	only marks,
	mark=o,mark size=1.8,
	scatter/use mapped color={
	  draw=green!50!black
	}
	]%
	table [col sep=comma,x=#3,y=#5] {#1};
	\addplot[
	scatter,
	only marks,
	mark=x,
	scatter/use mapped color={
	  draw=red
	}
	]%
	table [col sep=comma,x=#3,y=#5] {#2};
	\addplot[no marks] coordinates {(0.01,0.01) (9000,9000) };
	\addplot[no marks, densely dotted] coordinates {(0.01,0.1) (900,9000)};
	\addplot[no marks, densely dotted] coordinates {(0.1,0.01) (9000,900)};
	\end{axis}
	\end{tikzpicture}
}

\newcommand{\positivedatalocation}{resources/formatted_data.csv}
\newcommand{\negativedatalocation}{resources/formatted_data2.csv}
\newcommand{\benchmark}[4]{%
  \scatterplot{\positivedatalocation}{\negativedatalocation}{#1}{#2}{#3}{#4}
}

\DeclareRobustCommand{\good}{\Simley{0.5}{0.2}\xspace}
\DeclareRobustCommand{\bad}{\Simley{-0.5}{0.2}\xspace}
\newcommand{\Simley}[2]{%
 \begin{tikzpicture}[scale=#2]
 \newcommand*{\SmileyRadius}{1.0}%
 ;
 
 \pgfmathsetmacro{\eyeX}{0.5*\SmileyRadius*cos(30)}
 \pgfmathsetmacro{\eyeY}{0.5*\SmileyRadius*sin(30)}
 \draw [line width=0.25mm] (\eyeX-0.25,\eyeY) -- (\eyeX-0.25,\eyeY+0.375);
 \draw [line width=0.25mm] (-\eyeX+0.25,\eyeY) -- (-\eyeX+0.25,\eyeY+0.375);
 
 \pgfmathsetmacro{\xScale}{2*\eyeX/180}
 \pgfmathsetmacro{\yScale}{1.0*\eyeY}
 \draw[line width=0.25mm, domain=-\eyeX:\eyeX]
 plot ({\x},{
  -0.1+#1*0.15 %
  -#1*1.75*\yScale*(sin((\x+\eyeX)/\xScale))-\eyeY});
 \end{tikzpicture}%
}%

\newcolumntype{L}[1]{>{\raggedright\let\newline\\\arraybackslash\hspace{0pt}}m{#1}}
\newcolumntype{C}[1]{>{\centering\let\newline\\\arraybackslash\hspace{0pt}}m{#1}}
\newcolumntype{R}[1]{>{\raggedleft\let\newline\\\arraybackslash\hspace{0pt}}m{#1}}

\ifdefined\RESTATABLEFIX
\makeatletter

\def\thmt@rst@storecounters#1{%
\vspace{-1.9ex}%
  \bgroup
  \def\@currentlabel{}%
  \@for\thmt@ctr:=\thmt@innercounters\do{%
    \thmt@sanitizethe{\thmt@ctr}%
    \protected@edef\@currentlabel{%
      \@currentlabel
      \protect\def\@xa\protect\csname the\thmt@ctr\endcsname{%
        \csname the\thmt@ctr\endcsname}%
      \ifcsname theH\thmt@ctr\endcsname
        \protect\def\@xa\protect\csname theH\thmt@ctr\endcsname{%
          (restate \protect\theHthmt@dummyctr)\csname theH\thmt@ctr\endcsname}%
      \fi
      \protect\setcounter{\thmt@ctr}{\number\csname c@\thmt@ctr\endcsname}%
    }%
  }%
  \label{thmt@@#1@data}%
  \egroup
}%

\makeatother
\fi

\tikzset{
  ->/.style={-{Stealth[length=2mm, width=1mm]}},
}

\begin{document}

\maketitle

\begin{abstract}
  Multiple-environment MDPs (MEMDPs) capture finite sets of MDPs that share the states but differ in the transition dynamics.
	These models form a proper subclass of partially observable MDPs (POMDPs).
  We consider the synthesis of policies that robustly satisfy an almost-sure reachability property in MEMDPs, that is, \emph{one} policy that satisfies a property \emph{for all} environments.
	For POMDPs, deciding the existence of robust policies is an EXPTIME-complete problem.
	We show that this problem is PSPACE-complete for MEMDPs, while the policies require exponential memory in general.
  We exploit the theoretical results to develop and implement an algorithm that shows promising results in synthesizing robust policies for various benchmarks.
\end{abstract}

\section{Introduction}
Markov decision processes (MDPs) are the standard formalism to model sequential decision making under uncertainty.
A typical goal is to find a policy that satisfies a temporal logic specification~\cite{DBLP:books/daglib/0020348}.
Probabilistic model checkers such as \storm~\cite{DBLP:journals/sttt/HenselJKQV22} and \prism~\cite{KNP11} efficiently compute such policies.
A concern, however, is the robustness against potential perturbations in the environment. 
MDPs cannot capture such uncertainty about the shape of the environment.

Multi-environment MDPs (MEMDPs)~\cite{DBLP:conf/fsttcs/RaskinS14,DBLP:conf/aips/ChatterjeeCK0R20} contain a set of MDPs, called environments, over the same state space. 
The goal in MEMDPs is to find a single policy that satisfies a given specification in \emph{all} environments.
MEMDPs are, for instance, a natural model for MDPs with unknown system dynamics, where several domain experts provide their interpretation of the dynamics~\cite{DBLP:conf/aaai/ChadesCMNSB12}. These different MDPs together form a MEMDP.
MEMDPs also arise in other domains:
The guessing of a (static) password is a natural example in security.
In robotics, a MEMDP captures unknown positions of some static obstacle. 
One can interpret MEMDPs as a (disjoint) union of MDPs in which an agent only has partial observation, i.e., every MEMDP can be cast into a linearly larger partially observable MDP (POMDP)~\cite{DBLP:journals/ai/KaelblingLC98}. 
Indeed, some famous examples for POMDPs are in fact MEMDPs, such as \emph{RockSample}~\cite{DBLP:conf/uai/SmithS05} and \emph{Hallway}~\cite{DBLP:conf/icml/LittmanCK95}.
Solving POMDPs is notoriously hard~\cite{DBLP:journals/ai/MadaniHC03}, and thus, it is worthwhile to investigate natural subclasses.

We consider \emph{almost-sure specifications} where the probability needs to be one to reach a set of target states.
In MDPs, it suffices to consider memoryless policies.
Constructing such policies can be efficiently implemented by means of a  graph-search~\cite{DBLP:books/daglib/0020348}.
For MEMDPs, 
we consider the following problem:
\begin{center}
\emph{
	Compute \emph{one} policy that almost-surely reaches the target in \emph{all} environments.
	}
\end{center}
Such a policy robustly satisfies an almost-sure specification for a set of MDPs.

\paragraph{Our approach.}
Inspired by work on POMDPs, we construct a belief-observation MDP (BOMDP)~\cite{DBLP:journals/jcss/ChatterjeeCT16} that tracks the states of the MDPs and the (support of the) belief over potential environments.
We show that a policy satisfying the almost-sure property in the BOMDP also satisfies the property in the MEMDP. 
 
Although the BOMDP is exponentially larger than the MEMDP, we exploit its particular structure to create a PSPACE algorithm to decide whether such a robust policy exists.
The essence of the algorithm is a recursive construction of a fragment of the BOMDP, restricted to a setting in which the belief-support is fixed. Such an approach is possible, as the belief in a MEMDP behaves monotonically: Once we know that we are not in a particular environment, we never lose this knowledge.
This behavior is in contrast to POMDPs, where there is no monotonic behavior in belief-supports. 
The difference is essential: Deciding almost-sure reachability in POMDPs is EXPTIME-complete~\cite{DBLP:journals/jcss/Reif84,de1999verification}.
In contrast, the problem of deciding whether a policy for almost-sure reachability in a MEMDP exists is indeed PSPACE\emph{-complete}. 
We show the hardness using a reduction from the \emph{true quantified Boolean formula problem}. 
Finally, we cannot hope to extract a policy with such an algorithm, as the smallest policy for MEMDPs may require exponential memory in the number of environments. 

The PSPACE algorithm itself recomputes many results. 
For practical purposes, we create an algorithm that iteratively explores parts of the BOMDP. 
The algorithm additionally uses the MEMDP structure to generalize the set of states from which a winning policy exists and deduce efficient heuristics for guiding the exploration. 
The combination of these ingredients leads to an efficient and competitive prototype on top of the model checker \storm.

\medskip\noindent\textbf{Related work. } We categorize related work in three areas.

\smallskip\noindent\emph{MEMDPs.}
Almost-sure reachability for MEMDPs for exactly two environments has been studied by~\cite{DBLP:conf/fsttcs/RaskinS14}. 
We extend the results to arbitrarily many environments. 
This is nontrivial: For two environments, the decision problem has a polynomial time routine~\cite{DBLP:conf/fsttcs/RaskinS14}, whereas we show that the problem is PSPACE-complete for an arbitrary number of environments.
MEMDPs and closely related models such as hidden-model MDPs, hidden-parameter MDPs, multi-model MDPs, and concurrent MDPs~\cite{DBLP:conf/aaai/ChadesCMNSB12,DBLP:conf/qest/ArmingBCKS18,DBLP:journals/iiset/SteimleKD21,DBLP:journals/mmor/BuchholzS19} have been considered for quantitative properties\footnote{Hidden-parameter MDPs are different than MEMDPs in that they assume a prior over MDPs.
However, for almost-sure properties, this difference is irrelevant.}. 
The typical approach is to consider approximative algorithms for the undecidable problem in POMDPs~\cite{DBLP:conf/aips/ChatterjeeCK0R20} or adapt reinforcement learning algorithms~\cite{DBLP:conf/nips/AzarLB13,DBLP:journals/corr/abs-2111-09794}. These approximations are not applicable to almost-sure properties.

\smallskip\noindent\emph{POMDPs.}
One can build an underlying potentially infinite belief-MDP~\cite{DBLP:journals/ai/KaelblingLC98} that corresponds to the POMDP -- using model checkers~\cite{DBLP:journals/rts/Norman0Z17,DBLP:conf/atva/BorkJKQ20,DBLP:conf/tacas/BorkKQ22} to verify this MDP can answer the question for MEMDPs. 
For POMDPs, almost-sure reachability is decidable in exponential time~\cite{DBLP:journals/jcss/Reif84,de1999verification} via a construction similar to ours. Most qualitative properties beyond almost-sure reachability are undecidable~\cite{DBLP:journals/jacm/BaierGB12,DBLP:conf/csl/ChatterjeeCT13}.
Two dedicated algorithms that limit the search to policies with small memory requirements and employ a SAT-based approach~\cite{DBLP:conf/aaai/ChatterjeeCD16,DBLP:conf/cav/JungesJS21} to this NP-hard problem~\cite{de1999verification} are implemented in~\storm. We use them as baselines.

\smallskip\noindent\emph{Robust models.}
The high-level representation of MEMDPs is structurally similar to featured MDPs~\cite{DBLP:journals/fac/ChrszonDKB18,DBLP:conf/cav/AndriushchenkoC21} that represent sets of MDPs.
The proposed techniques are called family-based model checking and compute policies for every MDP in the family, whereas we aim to find one policy for all MDPs.
Interval MDPs~\cite{DBLP:conf/lics/JonssonL91,DBLP:journals/mor/WiesemannKR13,DBLP:conf/isola/Jaeger0BLJ20} and SGs~\cite{SGs} do not allow for dependencies between states and thus cannot model features such as various obstacle positions.
Parametric MDPs~\cite{DBLP:conf/qest/ArmingBCKS18,DBLP:conf/concur/WinklerJPK19,DBLP:journals/corr/abs-2207-06801} assume controllable uncertainty and do not consider robustness of policies.

\medskip\noindent\textbf{Contributions. }
We establish PSPACE-completeness for deciding almost-sure reachability in MEMDPs and show that the policies  may be exponentially large. 
Our iterative algorithm, which is the first specific to almost-sure reachability in MEMDPs,  builds fragments of the BOMDP. 
An empirical evaluation shows that the iterative algorithm  outperforms approaches dedicated to POMDPs. %

\section{Problem Statement}
\label{sec:problem_statement}

In this section, we provide some background and formalize the problem statement.%

For a set $X$, $\dist(X)$ denotes the set of probability distributions over $X$.
For a given distribution $d \in \dist(X)$, we denote its support as $\supp(d)$.
 For a finite set $X$, let $\unif(X)$ denote the uniform distribution. 
 $\dirac(x)$ denotes the Dirac distribution on $x\in X$.
We use short-hand notation for functions and distributions, $f = [ x \mapsto a, y \mapsto b ]$ means that $f(x) = a$ and $f(y) = b$.
We write $\powerset{X}$ for the powerset of $X$. %
For $n \in \mathbb{N}$ we write $[n] = \{ i \in \mathbb{N} \mid 1 \le i \le n \}$.

\newcommand{\pathbelief}{\mathcal{B}}
\newcommand{\stateset}{S}
\newcommand{\actionset}{A}
\newcommand{\transitionfuncs}{\{p_i\}_{i \in I}}
\newcommand{\initdist}{\iota_{\text{init}}}
\begin{definition}[MDP]
  A \emph{Markov Decision Process} is a tuple $\mdp = \tuple{\stateset, \actionset, \initdist, p}$ where $\stateset$ is the finite set of states, $\actionset$ is the finite set of actions, $\initdist \in \dist(\stateset)$ is the initial state distribution, and $p\colon \stateset \times \actionset \to \dist(\stateset)$ is the transition function.
\end{definition}
The transition function is total, that is, for notational convenience MDPs are \emph{input-enabled}.
This requirement does not affect the generality of our results.
A \emph{path} of an MDP is a sequence
$\pi = s_{0}a_{0}s_1a_{1}\ldots s_{n}$ such that $\initdist(s_{0}) > 0$ and $p(s_{i},a_{i})(s_{i+1}) > 0$ for all $0 \le i < n$.
The last state of $\pi$ is $\last(\pi)=s_n$.
The set of all finite paths is $\Path$
  and $\Path(S')$
denotes the paths starting in a state from $S'\subseteq \stateset$.
The set of \emph{reachable states} from $S'$ is $\reachable{S'}$.
If $S' = \supp(\initdist)$ we just call them \emph{the} reachable states.
The MDP restricted to reachable states from a distribution $d \in \dist(S)$ is $\onlyreachable{\mdp}{d}$, where $d$ is the new initial distribution.
A state $s\in S$ is \emph{absorbing} if $\reachable{\{s\}}=\{s\}$.
An MDP is acyclic, if each state is absorbing or not reachable from its successor states.

Action choices are resolved by a \emph{policy} $\sigma\colon \Path \to \dist(\actionset)$ that maps paths to distributions over actions.
A policy of the form $\sigma \colon \stateset \to \dist(\actionset)$ is called \emph{memoryless}, \emph{deterministic} if we have $\sigma\colon \Path \to \actionset$; and, \emph{memoryless deterministic} for $\sigma \colon \stateset \to \actionset$.
For an MDP $\mdp$, we denote the probability of a policy $\sigma$ reaching some target set $T \subseteq \stateset$ starting in state $s$ as ${\Pr}_{\mdp}(s \rightarrow T \mid \sigma)$. 
More precisely, ${\Pr}_{\mdp}(s \rightarrow T \mid \sigma)$ denotes the probability of all paths from $s$ reaching $T$ under $\sigma$. 
We use ${\Pr}_{\mdp}(T \mid \sigma)$ if $s$ is distributed according to $\initdist$.
\begin{definition}[MEMDP]
  A \emph{Multiple Environment MDP} is a tuple $\memdp = \tuple{\stateset, \actionset, \initdist, \transitionfuncs}$ with $\stateset, \actionset, \initdist$ as for MDPs, and $\transitionfuncs$ is a \emph{set of transition functions}, where $I$ is a finite set of \emph{environment} indices.
\end{definition}
Intuitively, MEMDPs form sets of MDPs (environments) that share states and actions, but differ in the transition probabilities.
For MEMDP $\memdp$ with index set $I$ and a set $I' \subseteq I$, we define the restriction of environments as the MEMDP $\restrictenv{\memdp}{I'} = \tuple{\stateset, \actionset, \initdist, \{p_{i}\}_{i \in I'}}$.
Given an environment $i \in I$, we denote its corresponding MDP as $\memdp_{i}=\tuple{S, A, \initdist, p_{i}}$.
A MEMDP with only one environment is an MDP.
Paths and policies are defined on the states and actions of MEMDPs and do not differ from MDP policies. A MEMDP is acyclic, if each MDP is acyclic.
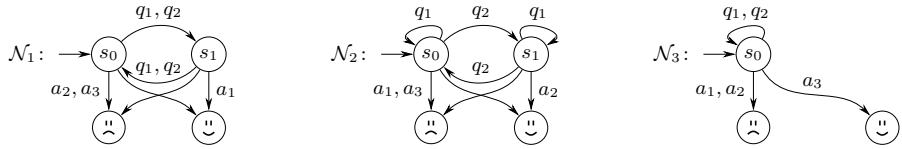
\begin{figure}[t]
  \centering
  \resizebox{\linewidth}{!}{
    \begin{tikzpicture}
  \begin{scope}
  \node (si) {$\memdp_{1}\colon$};
  \node[circle,draw,inner sep=2pt] (s0) [right=of si, xshift=-0.5cm] {$s_{0}$};
  \node[circle,draw,inner sep=2pt] (s1) [right=of s0] {$s_{1}$};

  \draw[->] (si) -- (s0);

  \draw[->] (s0) to[out=45, in=135] node[above]{$q_{1}, q_{2}$} (s1);
  \draw[->] (s1) to[out=-135, in=-45] node[above]{$q_{1}, q_{2}$} (s0);

  \node[align=center,circle,draw,inner sep=2pt] (frowney) [below=0.6cm of s0] {$\bad$}; 
  \node[align=center,circle,draw,inner sep=2pt] (smiley) [below=0.6cm of s1] {$\good$};	

  \draw[->] (s0) to[out=-60, in=135] (smiley);%
  \draw[->] (s1) to node[right]{$a_{1}$} (smiley);

  \draw[->] (s0) to node[left]{$a_{2},a_{3}$} (frowney);
  \draw[->] (s1) to[out=-120, in=45] (frowney);
  \end{scope}

  \begin{scope}[xshift=5cm]
  \node (si) {$\memdp_{2}\colon$};
  \node[circle,draw,inner sep=2pt] (s0) [right=of si, xshift=-0.5cm] {$s_{0}$};
  \node[circle,draw,inner sep=2pt] (s1) [right=of s0] {$s_{1}$};

  \draw[->] (si) -- (s0);
  \draw[->] (s0) to[out=45, in=135] node[above]{$q_{2}$} (s1);
  \draw[->] (s1) to[out=-135, in=-45] node[above]{$q_{2}$} (s0);

  \draw[->] (s0) to[out=60, in=150, looseness=3] node[above]{$q_{1}$} (s0);
  \draw[->] (s1) to[out=120, in=30, looseness=3] node[above]{$q_{1}$} (s1);

  \node[align=center,circle,draw,inner sep=2pt] (frowney) [below=0.6cm of s0] {$\bad$}; 
  \node[align=center,circle,draw,inner sep=2pt] (smiley) [below=0.6cm of s1] {$\good$}; 

  \draw[->] (s0) to[out=-60, in=135] (smiley);%
  \draw[->] (s1) to node[right]{$a_{2}$} (smiley);

  \draw[->] (s0) to node[left]{$a_{1},a_{3}$} (frowney);
  \draw[->] (s1) to[out=-120, in=45] (frowney);
  \end{scope}

  \begin{scope}[xshift=10cm]
  \node (si) {$\memdp_{3}\colon$};
  \node[circle,draw,inner sep=2pt] (s0) [right=of si, xshift=-0.5cm] {$s_{0}$};

  \draw[->] (si) -- (s0);
  \draw[->] (s0) to[out=60, in=150, looseness=3] node[above]{$q_{1}, q_{2}$} (s0);

  \node[align=center,circle,draw,inner sep=2pt] (frowney) [below=0.6cm of s0] {$\bad$}; 
  \node[align=center,circle,draw,inner sep=2pt] (smiley) [below=0.6cm of s0, xshift=2.0cm] {$\good$};

  \draw[->] (s0) to[out=-60, in=135] node[above]{$a_{3}$} (smiley);
  \draw[->] (s0) to node[left]{$a_{1},a_{2}$} (frowney);
  \end{scope}
\end{tikzpicture}
  }
  \caption{Example MEMDP}
  \label{fig:example_memdp}
\end{figure}

\begin{example}
\label{ex:runningmemdp}
  \cref{fig:example_memdp} shows an MEMDP with three environments $\memdp_i$.
  An agent can ask two questions, $q_{1}$ and $q_{2}$. The response is either  `switch' ($s_{1} \leftrightarrow s_{2}$), or `stay' (loop).
In $\memdp_1$, the response to $q_1$ and $q_2$ is to switch. 
  In $\memdp_2$, the response to $q_1$ is stay, and to $q_2$ is switch. 
  The agent can guess the environment using $a_{1}, a_{2}, a_{3}$. 
  Guessing $a_i$ leads to the target $\{\good\}$ only in environment $i$.
  Thus, an agent must deduce the environment via $q_{1}, q_{2}$ to surely reach the target.
  \exampleend
\end{example}

\begin{definition}[Almost-Sure Reachability]
  An almost-sure reachability property is defined by a set $T \subseteq S$ of target states. 
  A policy $\sigma$ satisfies the property $T$ for MEMDP $\memdp = \tuple{S, A, \initdist, \{p_{i}\}_{i \in I}}$ iff $\forall i \in I\colon {\Pr}_{\memdp_{i}}(T \mid \sigma)=1$.
\end{definition}
In other words, a policy $\sigma$ satisfies an almost-sure reachability property $T$, called \emph{winning}, if and only if the probability of reaching $T$ \emph{within each MDP} is one. 
By extension, a state $s \in S$ is winning if there exists a winning policy when starting in state $s$. 
Policies and states that are not winning are losing.

\noindent
We will now define both the decision and policy problem:
\begin{mdframed}
  Given a MEMDP $\memdp$ and an almost-sure reachability property $T$.\\ The \textbf{Decision Problem} asks to decide if a policy  exists that satisfies $T$. \\ The \textbf{Policy Problem} asks to compute such a policy, if it exists.
\end{mdframed}
In \cref{sec:complexity} we discuss the computational complexity of the decision problem. Following up, in \cref{sec:algorithm} we present our algorithm for solving the policy problem. Details on its implementation and evaluation will be presented in \cref{sec:implementation_experiments}.

\newcommand{\observations}{Z}
\newcommand{\obsfun}{O}
\newcommand{\obs}{z}
\newcommand{\pomdp}{\mathcal{G}}
\newcommand{\union}[1]{{#1}_{\sqcup}}
\newcommand{\custombomdp}[1]{\mathcal{G}_{#1}}
\newcommand{\bomdp}{\custombomdp{\memdp}}
\newcommand{\localbomdp}[1]{\textsc{Loc}\mathcal{G}({#1})}
\newcommand{\outdated}[1]{\textcolor{orange}{#1}}
\section{A Reduction To Belief-Observation MDPs}
In this section, we reduce the policy problem, and thus also the decision problem, to finding a policy in an exponentially larger belief-observation MDP. 
This reduction is an elementary building block for the construction of our PSPACE algorithm and the practical implementation.
Additional information such as proofs for statements throughout the paper are available in the technical report~\cite{technicalreport}.

\subsection{Interpretation of MEMDPs as Partially Observable MDPs} \label{sec:pomdps}
\begin{definition}[POMDP]
	A partially observable MDP (POMDP) is a tuple $\tuple{\mdp, \observations, \obsfun}$ %
	with an MDP $\mdp = \tuple{S, A, \initdist, p}$,
	  a set $\observations$ of \emph{observations}, and an \emph{observation function}
	 $\obsfun\colon S \rightarrow \observations$. 
\end{definition}
A POMDP is an MDP where states are labelled with observations. 
We lift $\obsfun$ to paths and use $\obsfun(\pi) = \obsfun(s_1) a_{1} \obsfun(s_{2}) \hdots \obsfun(s_n)$. We use observation-based policies $\sigma$, i.e., policies s.t.\ for $\pi, \pi' \in \Path$,
$
 \obsfun(\pi) = \obsfun(\pi')  \text{ implies } \sigma(\pi) = \sigma(\pi'). 
 $
A MEMDP can be cast into a POMDP that is made up as the disjoint union: 
\begin{definition}[Union-POMDP]
\label{def:unionpomdp}
  Given an MEMDP $\memdp = \tuple{\stateset, \actionset, \initdist, \transitionfuncs}$ we define its \emph{union-POMDP} $\union{\memdp} = \tuple{\tuple{\stateset', \actionset, \initdist', p'}, \observations, \obsfun}$, with states $\stateset' = S \times I$, initial distribution $\initdist'(\tuple{s, i}) = \initdist(s) \cdot |I|^{-1}$, transitions  $p'(\tuple{s, i}, a)(\tuple{s', i}) = p_{i}(s, a)(s')$,
    observations $\observations = S$, and observation function $\obsfun(\tuple{s, i}) = s$.
\end{definition}%
A policy may observe the state $s$ but not in which MDP we are. This forces any observation-based policy to take the same choice in all environments.
\begin{restatable}{lemma}{unionpomdp} \label{thm:unionpomdp}
	Given MEMDP $\memdp$, there exists a winning policy iff there exists an observation-based policy $\sigma$ such that ${\Pr}_{\union{\memdp}}(T \mid \sigma)=1$.
\end{restatable}
\noindent
The statement follows as, first, any observation-based policy of the POMDP can be applied to the MEMDP, second, vice versa, any MEMDP policy is observation-based, and third, the induced MCs under these policies are isomorphic.

\newcommand{\update}[4]{\mathsf{Up}(#1, #2, #3, #4)}

\subsection{Belief-observation MDPs}
For POMDPs, memoryless policies are not sufficient, which makes computing policies intricate. We therefore add the information that the history --- i.e., the path until some point --- contains.
In MEMDPs, this information is the \emph{(environment-)belief (support)} $J \subseteq I$, as the set of environments that are consistent with a path in the MEMDP. Given a belief $J \subseteq I$ and a state-action-state transition $s \xrightarrow{a} s'$, then we define $\update{J}{s}{a}{s'} = \{ i \in J \mid p_{i}(s, a,s') > 0\}$, i.e., the subset of environments in which the transition exists.
For a path $\pi \cdot s \in \Path$, we define its corresponding belief $\pathbelief(\pi) \subseteq I$ recursively as:\begin{align*}
  \pathbelief(s_{0}) &= I \quad\text{ and }\quad
  \pathbelief(\pi \cdot  s a s') = \update{\pathbelief(\pi \cdot s)}{s}{a}{s'}
\end{align*}
The belief in a MEMDP monotonically decreases along a path, i.e., if we know that we are not in a particular environment, this remains true indefinitely.

We aim to use a model where memoryless policies suffice. To that end, we cast MEMDPs into the exponentially larger belief-observation MDPs~\cite{DBLP:journals/jcss/ChatterjeeCT16}\footnote{This translation is notationally simpler than going via the union-POMDP.}.
\begin{definition}[BOMDP]
  For a MEMDP $\memdp = \tuple{\stateset, \actionset, \initdist, \transitionfuncs}$, we define its \emph{belief-observation MDP} (BOMDP) as a POMDP $\bomdp = \tuple{\tuple{\stateset', \actionset, \initdist', p'}, \observations, \obsfun}$ with states $\stateset' = S \times I \times \powerset{I}$, initial distribution
    $\initdist'(\tuple{s, j, I}) = \initdist(s) \cdot |I|^{-1}$, transition relation $p'(\tuple{s, j, J}, a)(\tuple{s', j, J'}) = p_{j}(s, a, s')$ with $J' = \update{J}{s}{a}{s'}$, observations $\observations = S \times \powerset{I}$, and observation function $\obsfun(\tuple{s, j, J}) = \tuple{s, J}$.
\end{definition}
Compared to the union-POMDP, BOMDPs also track the belief by updating it accordingly.
We clarify the correspondence between paths of the BOMDP and the MEMDP. 
For a path $\pi$ through the MEMDP,  we can mimic this path exactly in the MDPs $\memdp_j$ for $j \in \pathbelief(\pi)$.
As we track $\pathbelief(\pi)$ in the state, we can deduce from the BOMDP state in which environments we can be.
\begin{restatable}{lemma}{pathslemma} \label{lem:paths}
  For MEMDP $\memdp$ and the path $\tuple{s_1,j,J_1}a_1\tuple{s_2,j,J_2} \hdots \tuple{s_n,j,J_n}$ of the BOMDP $\bomdp$, let $j \in J_{1}$. Then: $J_n \neq \emptyset$ and 
  the path $s_1a_1 \hdots s_n$ exists in MDP $\memdp_i$ iff $i \in J_1 \cap J_n$.
\end{restatable} \noindent
Consequently, the belief of a path can be uniquely determined by the observation of the last state reached, hence the name belief-observation MDPs.
\begin{restatable}{lemma}{beliefobservation} \label{lem:beliefobservation}
For every pair of  paths $\pi, \pi'$ in a BOMDP, we have:
  \[ \obsfun(\last(\pi)) = \obsfun(\last(\pi'))\quad\text{ implies }\quad \pathbelief(\pi) = \pathbelief(\pi'). \]
\end{restatable}\noindent
For notation, we define $S_J = \{ \tuple{s, j, J} \mid j \in J, s \in S \}$, and analogously write $\observations_J =  \{ \tuple{s, J} \mid s \in S \}$. We lift the target states $T$ to states in the BOMDP: $T_{\bomdp} = \{ \tuple{s, j, J} \mid s \in T, J \subseteq I, j \in J\}$ and define target observations $T_\observations = \obsfun(T_{\bomdp})$.

\begin{definition}[Winning in a BOMDP]
	Let $\bomdp$ be a BOMDP with target observations $T_\observations$.
	An observation-based policy $\sigma$ is winning from some observation $\obs \in \observations$, 
	if for all $s \in \obsfun^{-1}(\obs)$ 
	it holds that $\Pr_{\bomdp}( s \rightarrow \obsfun^{-1}(T_\observations) \mid \sigma) = 1$.
\end{definition} 
Furthermore, a policy $\sigma$ is \emph{winning} if it is winning for the initial distribution $\initdist$.
An \emph{observation $\obs$ is winning} if there exists a winning policy for $\obs$. 
The \emph{winning region} $\win{\bomdp}{T}$ is the set of all winning observations.

\noindent Almost-sure winning in the BOMDP corresponds to winning in the MEMDP. 
\begin{restatable}{theorem}{winningmemdpbomdp} \label{thm:winning_memdp_bomdp}
There exists a winning policy for a MEMDP $\memdp$ with target states $T$  iff there exists a winning policy in the BOMDP $\bomdp$ with target states $T_{\bomdp}$.
\end{restatable} \noindent
Intuitively, the important aspect is that for almost-sure reachability, observation-based memoryless policies are sufficient~\cite{DBLP:journals/ai/ChatterjeeCGK16}.
For any such policy, the induced Markov chains on the union-POMDP and the BOMDP are bisimilar~\cite{DBLP:journals/jcss/ChatterjeeCT16}.

BOMDPs make policy search conceptually easier.
First, as memoryless policies suffice for almost-sure reachability,
winning regions are independent of fixed policies: For policies $\sigma$ and $\sigma'$ that are winning in observation $\obs$ and $\obs'$, respectively, there must exist a policy $\hat{\sigma}$ that is winning for both $\obs$ and $\obs'$.
Second, winning regions can be determined in polynomial time in the size of the BOMDP~\cite{DBLP:journals/jcss/ChatterjeeCT16}.

\subsection{Fragments of BOMDPs}
To avoid storing the exponentially sized BOMDP, we only build  fragments: %
We may select any set of observations as \emph{frontier} observations and make the states with those observations absorbing.  We later discuss the selection of frontiers.
\begin{definition}[Sliced BOMDP]
  For a BOMDP $\bomdp = \tuple{\tuple{S, A, \initdist, p}, \observations, \obsfun}$ and a set of \emph{frontier observations} $F \subseteq \observations$, we define a BOMDP $\cut{\bomdp}{F} = \tuple{\tuple{S, A, \initdist, p'}, \observations, \obsfun}$ with:
  \vspace{-0.7em}
  \[
    \forall s \in S, a \in A\colon p'(s, a) =
    \begin{cases}
      \dirac(s) & \text{if } \obsfun(s) \in F,\\
      p(s,a) & \text{otherwise.}\\
    \end{cases}
  \]
  \vspace{-1.2em}
  \end{definition}
We exploit this sliced BOMDP to derive constraints on the set of winning states.
\begin{restatable}{lemma}{winningfrontier}\label{lem:winningfrontier}
  For every BOMDP $\bomdp$ with states $S$ and targets $T$ and for all frontier observations $F \subseteq \observations$ it holds that:\ $\win{\cut{\bomdp}{F}}{T} \subseteq \win{\bomdp}{T} \subseteq \win{\cut{\bomdp}{F}}{T \cup F}$.
\end{restatable} \noindent
Making (non-target) observations absorbing extends the set of losing observations, while adding target states extends the set of winning observations.

\section{Computational Complexity}
\label{sec:complexity}
The BOMDP $\bomdp$ above yields an exponential time \emph{and space} algorithm via  \cref{thm:winning_memdp_bomdp}.
We can avoid the exponential memory requirement.
This section shows the PSPACE-completeness of deciding whether a winning policy exists.
\begin{theorem}
  The almost-sure reachability decision problem is PSPACE-complete.
\end{theorem}
 The result follows from Lemmas~\ref{lem:decision_pspace_hard} and \ref{lem:pspacealg} below. In \cref{sec:policy_problem}, we show that representing the winning policy itself may however require exponential space.

\subsection{Deciding Almost-Sure Winning for MEMDPs in PSPACE}
\label{sec:decision_problem}

We develop an algorithm with a polynomial memory footprint. 
The algorithm exploits locality of cyclic behavior in the BOMDP, as formalized by  an acyclic \emph{environment graph} and \emph{local BOMDPs} that match the nodes in the environment graph.
The algorithm recurses on the environment graph while memorizing results from polynomially many local BOMDPs.

\subsubsection{The graph-structure of BOMDPs.}
First, along a path of the MEMDP, we will only gain information and are thus able to rule out certain environments~\cite{DBLP:conf/aips/ChatterjeeCK0R20}.
Due to the monotonicity of the update operator, we have for any BOMDP that $\tuple{s,j,J} \in \reachable{\tuple{s', j, J'}}$ implies $J \subseteq J'$.
We define a graph over environment sets that describes how the belief-support can update over a run.
\begin{figure}[t]
\centering
\begin{tikzpicture}
	\node[rectangle,draw] (111) {\scriptsize $\{1,2,3\}$};
	\node[rectangle,draw] (011) at (2,-0.66) {\scriptsize $\{2,3\}$};
	\node[rectangle,draw] (110) at (2,0.66) {\scriptsize $\{1,2\}$};
	\node[rectangle,draw] (100) at (4,0.66) {\scriptsize $\{1\}$};
	\node[rectangle,draw] (010) at (4,0) {\scriptsize $\{2\}$};
	\node[rectangle,draw] (001) at (4,-0.66) {\scriptsize $\{3\}$};
	\node[rectangle,draw] (101) at (6,0) {\scriptsize $\{1,3\}$};

	\draw[->] (111) -- (011);
	\draw[->] (111) -- (110);
	\draw[->] (111) -- (100);
	\draw[->] (111) -- (010);
	\draw[->] (111) -- (001);

	\draw[->] (110) -- (100);
	\draw[->] (110) -- (010);
	\draw[->] (011) -- (010);
	\draw[->] (011) -- (001);

	\draw[->] (101) -- (100);
	\draw[->] (101) -- (001);
\end{tikzpicture}
\caption{The environment graph for our running example.}
\label{fig:envgraphex}
\end{figure}
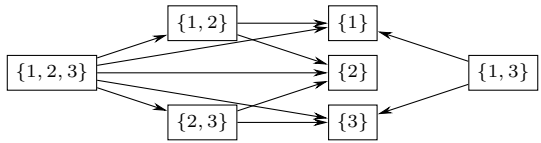%
\begin{definition}[Environment graph]\label{def:envgraph}
Let $\memdp$ be a MEMDP and $p$ the transition function of $\bomdp$.
The \emph{environment graph} $\egraph=(\egraphV,\egraphE)$ for $\memdp$ is a directed graph with vertices $\egraphV=\powerset{I}$ and edges 
\[\egraphE = \{ \tuple{J, J'} \mid \exists s, s' \in S, a \in A, j \in I. p(\tuple{s,j,J}, a, \tuple{s', j, J'}) > 0 \text{ and } J \neq J' \}. \]
\end{definition}%
\begin{example}
\cref{fig:envgraphex} shows the environment graph for the MEMDP in Ex.~\ref{ex:runningmemdp}. It consists of the different belief-supports. For example, the transition from $\{1,2,3\}$ to $\{2,3\}$ and to $\{1\}$ is due to the action $q_1$ in state $s_0$, as shown in Fig.~\ref{fig:example_memdp}.%
\exampleend
\end{example}%
Paths in the environment graph abstract paths in the BOMDP. Path fragments where the belief-support remains unchanged are summarized into one step, as we do not create edges of the form $\tuple{J,J}$. We formalize this idea:
Let $\pi=\tuple{s_1, j, J_1}a_1\tuple{s_2, j, J_2}\dots\tuple{s_n, j, J_n}$ be a path in the BOMDP. 
For any  $J \subseteq I$, we call $\pi$ a \emph{$J$-local path}, if $J_i = J$ for all $i \in [n]$.
\begin{restatable}{lemma}{localpath} \label{lem:localpath}
  For a MEMDP $\memdp$ with environment graph $\egraph$, there is a path $J_1 \dots J_n$ if there is a path $\pi = \pi_1\dots\pi_n$ in  $\bomdp$  s.t.\ every $\pi_i$ is $J_i$-local.
\end{restatable} \noindent
The shape of the environment graph is crucial for the algorithm we develop.
\begin{restatable}{lemma}{envgraphshape} \label{lem:envgraphshape}
Let $\egraph = (\egraphV, \egraphE)$ be an environment graph for MEMDP $\memdp$.
First, $\egraphE(J,J')$ implies $J' \subsetneq J$. Thus, $G$ is acyclic and has maximal path length $|I|$. The maximal outdegree of the graph is $|S|^2|A|$.
\end{restatable} \noindent
The monotonicity regarding $J, J'$ follows from definition of the belief update.
The bound on the outdegree is a consequence from Lemma~\ref{lem:locbsgsize} below.

\subsubsection{Local belief-support BOMDPs.}
Before we continue, we remark that the (future) dynamics in a BOMDP only depend on the current state and set of environments. More formally, we capture this intuition as follows.
\begin{restatable}{lemma}{restrictedbsg} \label{lem:restrictedbsg}
Let $\bomdp$ be a BOMDP with states $S'$. 
For any state $\tuple{s, j, J} \in S'$, let 
$\memdp' = \onlyreachable{\restrictenv{\memdp}{J}}{\dirac(s)}$ and
$Y = \{ \tuple{s, i, J} \mid i \in J\}$.
Then: \[ \onlyreachable{\bomdp}{\unif(Y)} = \custombomdp{\memdp'}.\]
\end{restatable} \noindent
The key insight is that restricting the MEMDP does not change the transition functions for the environments $j \in J$.
Furthermore, using monotonicity of the update, we only reach BOMDP-states whose behavior is determined by the environments in $J$.

This intuition allows us to analyze the BOMDP locally and lift the results to the complete BOMDP. 
We define a local BOMDP as the part of a BOMDP starting in any state in  $S_J$.
All observations not in $\observations_J = S \times \{J\}$ are made absorbing.
\begin{definition}[Local BOMDP]
\label{def:localbsg}
Given a MEMDP $\memdp$ with BOMDP $\bomdp$ and a set of environments $J$. 
The \emph{local BOMDP} for environments $J$ is the fragment \[ \localbomdp{J} = \onlyreachable{\cut{\custombomdp{\restrictenv{\memdp}{J}}}{F}}{\unif(S_J)}\quad\text{ where }\quad F = \observations \setminus \observations_J\ . \]
\end{definition}
This definition of a local BOMDP coincides with a fragment of the complete BOMDP. 
We then mark exactly the winning observations restricted to the environment sets $J' \subsetneq J$ as winning in the local BOMDP and compute all winning observations in the local BOMDP. These observations are winning in the  complete BOMDP. The following concretization of \cref{lem:winningfrontier} formalizes this.
\begin{restatable}{lemma}{winninglocally}\label{lem:winninglocally}%
  Consider a MEMDP $\memdp$ and a subset of environments $J$.
  \[
    O\left(\win{\localbomdp{J}}{T'_{\bomdp}}\right) \cap  \observations_J ~=~  O\left(\win{\bomdp}{T_{\bomdp}}\right) \cap \observations_J \quad\text{with}\quad T'_{\bomdp} = T_{\bomdp} \cup (\win{\bomdp}{T_{\bomdp}}  \setminus S_J).
  \]
\end{restatable}%
\noindent Furthermore, local BOMDPs are polynomially bounded in the size of the MEMDP.
\begin{restatable}{lemma}{locbsgsize}\label{lem:locbsgsize}%
Let $\memdp$ be a MEMDP with states $S$ and actions $A$. 
$\localbomdp{J}$ has at most $\mathcal{O}(|S|^2\cdot|A|\cdot|J|)$ states and $\mathcal{O}(|S|^{2}\cdot|A|\cdot|J|^2)$ transitions%
\footnote{The number of transitions is the number of nonzero entries in $p$}.
\end{restatable}%

\subsubsection{A PSPACE algorithm.}
We present Algorithm~\ref{alg:search_algorithm} for the MEMDP \textbf{decision problem}, which recurses depth-first over the paths in the environment graph\footnote{In contrast to depth-first-search, we do not memorize nodes we visited earlier.}.
We first state the correctness and the space complexity of this algorithm.

\newcommand{\pspacealgname}{\textsc{ASWinning}}
\begin{algorithm}[t]
\begin{algorithmic}[1]
  \Function{Search}{MEMDP $\memdp = \tuple{S, A, {\{p_i\}}_{i \in I}, \initdist}$, $J \subseteq I$, $T \subseteq S$}
      \State $T' \gets \{ \tuple{s,j, J} \mid j \in J, s \in T \}$ 
      \For{\upshape$J'$ s.t.\ $\egraphE(J,J')$} \Comment{Consider the edges in the env.~graph (Def.~\ref{def:envgraph})}
          \State $W_{J'} \assign \textsc{Search}(\memdp, J', T)$ \Comment{Recursion!}
          \State $T' \assign T' \cup \{ \tuple{s, j, J'} \mid j \in J, \tuple{s, J'} \in W_{J'} \}$\label{line:deftprime}
      \EndFor
      \State \Return $\win{\localbomdp{J}}{T'} \cap \observations_J$\Comment{Construct BOMDP as in Def.~\ref{def:localbsg}, then model check}
  \EndFunction
\\
  \Function{\pspacealgname}{MEMDP $\memdp = \tuple{S, A, {\{p_i\}}_{i \in I}, \initdist}$, $T \subseteq S$}
      \State \Return $\obsfun(\supp(\initdist)) \subseteq \textsc{Search}(\memdp, I, T)$
  \EndFunction
\end{algorithmic}
\caption{Search algorithm}%
\label{alg:search_algorithm}%
\end{algorithm}

\begin{lemma}\label{lem:pspacealg}
	\pspacealgname{} in Alg.~\ref{alg:search_algorithm} solves the decision problem in PSPACE.
\end{lemma}
To prove correctness, we first note that $\textsc{Search}(\memdp, J, T)$ computes $\win{\bomdp}{T_{\bomdp}} \cap \observations_J$. 
We show this by induction over the structure of the environment graph. 
For all $J$ without outgoing edges, the local BOMDP coincides with a BOMDP just for environments $J$ (\cref{lem:restrictedbsg}). 
Otherwise, observe that $T'$ in line~\ref{line:deftprime} coincides with its definition in \cref{lem:winninglocally} and thus, by the same lemma, we return $\win{\bomdp}{T_{\bomdp}} \cap \observations_J$.
To finalize the proof, a winning policy exists in the MEMDP if the observation of the initial states of the BOMDP are winning (\cref{thm:winning_memdp_bomdp}).
The algorithm must terminate as it recurses over all paths of a finite acyclic graph, see~\cref{lem:envgraphshape}.
Following~\cref{lem:locbsgsize}, the number of frontier states is then bounded by $|S|^2\cdot|A|$. 
The main body of the algorithm therefore requires polynomial space, and the maximal recursion depth (stack height) is $|I|$ (\cref{lem:envgraphshape}).
Together, this yields a space complexity in $\mathcal{O}(|S|^{2}\cdot|A|\cdot|I|^2)$.

\subsection{Deciding Almost-Sure Winning for MEMDPs Is PSPACE-hard} It is not possible to improve the algorithm beyond PSPACE. 
\begin{restatable}{lemma}{pspacehard}\label{lem:decision_pspace_hard}
  The MEMDP decision problem is PSPACE-hard.
\end{restatable}
\noindent 
Hardness holds even for acyclic MEMDPs and uses the following fact.
\begin{restatable}{lemma}{detsuffices}\label{lem:acyc_det_suffices}
	If a winning policy exists for an acyclic MEMDP, there also exists a winning policy that is deterministic.
\end{restatable}
\noindent
In particular, almost-sure reachability coincides with avoiding the sink states. This is a safety property. For safety, deterministic policies are sufficient, as randomization visits only additional states, which is not beneficial for safety.

Regarding \Cref{lem:decision_pspace_hard}, we sketch a polynomial-time reduction from the PSPACE-complete TQBF problem~\cite{DBLP:books/fm/GareyJ79} problem to the MEMDP decision problem. 
  Let $\Psi$ be a QBF formula,
$
    \Psi = \exists x_{1} \forall y_{1} \exists x_{2} \forall y_{2} \ldots \exists x_{n} \forall y_{n}\big[\Phi\big]
 $
with $\Phi$ a Boolean formula in conjunctive normal form. %
The problem is to decide whether $\Psi$ is true.
\begin{figure}[t]
  \centering
  \resizebox{\linewidth}{!}{
    \begin{tikzpicture}
  \node[state,initial where=above] (x1) {$x$};
  \node[left=of x1, xshift=0.5cm] (x1i) {$\mathcal{M}_{1}\colon$};
  \node[state, right=of x1] (x1f) {$x\bot$};
  \node[state, above=0.6cm of x1f] (x1t) {$x\top$};

  \draw[->] (x1i) -- (x1);
  \node[right=of x1, xshift=-0.5cm, inner sep=0mm, outer sep=0mm] (x1m) {};
  \filldraw (x1m) circle (1pt);
  \draw (x1) -- node[above] {$\envact$}(x1m);
  \draw[->] (x1m) -- node[above, left] {$\frac12$} (x1t);
  \draw[->] (x1m) -- node[below] {$\frac12$} (x1f);

  \node[state, right=of x1f] (y1) {$y$};

  \node[state, right=of x1t, accepting] (w1) {$W$};
  \node[state, right=of y1] (f1) {$F$};

  \draw[->] (x1t) -- node[above]{$\envact$} (w1);
  \draw[->] (x1f) -- node[above]{$\envact$} (y1);
  \draw[->] (y1) -- node[right]{$\top$} (w1);
  \draw[->] (y1) -- node[above]{$\bot$} (f1);

  \draw[->] (w1) to[in=45, out=135, looseness=3] (w1);
  \draw[->] (f1) to[in=45, out=135, looseness=3] (f1);

  \node[state, right=of f1, initial where=above, xshift=1cm] (x2) {$x$};
  \node[left=of x2, xshift=0.5cm] (x2i) {$\mathcal{M}_{2}\colon$};
  \node[state, right=of x2] (x2t) {$x\top$};
  \node[state, above=0.6cm of x2t] (x2f) {$x\bot$};

  \draw[->] (x2i) -- (x2);
  \node[right=of x2, xshift=-0.5cm, inner sep=0mm, outer sep=0mm] (x2m) {};
  \filldraw (x2m) circle (1pt);
  \draw (x2) -- node[above] {$\envact$}(x2m);
  \draw[->] (x2m) -- node[below] {$\frac12$} (x2t);
  \draw[->] (x2m) -- node[above, left] {$\frac12$} (x2f);

  \node[state, right=of x2t] (y2) {$y$};

  \node[state, right=of x2f, accepting] (w2) {$W$};
  \node[state, right=of y2] (f2) {$F$};

  \draw[->] (x2t) -- node[above]{$\envact$} (y2);
  \draw[->] (x2f) -- node[above]{$\envact$} (w2);
  \draw[->] (y2) -- node[above]{$\top$} (f2);
  \draw[->] (y2) -- node[right]{$\bot$} (w2);

  \draw[->] (w2) to[in=45, out=135, looseness=3] (w2);
  \draw[->] (f2) to[in=45, out=135, looseness=3] (f2);
\end{tikzpicture}
  }

  \caption{Constructed MEMDP for the QBF formula $\forall x \exists y \big [ (x \vee y) \wedge (\neg x \vee \neg y) \big]$.}
  \label{fig:example_qbf}
\end{figure}
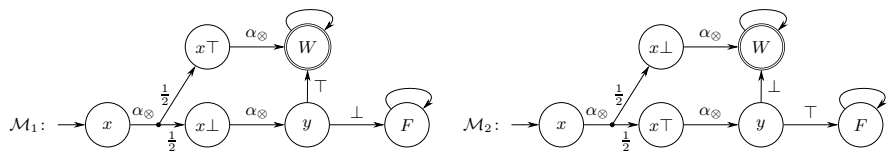
\begin{example}
  Consider the QBF formula
  $\Psi = \forall x \exists y \big [ (x \vee y) \wedge (\neg x \vee \neg y) \big]$.
  We construct a MEMDP with an environment for every clause, see \cref{fig:example_qbf}\footnote{We depict a slightly simplified MEMDP for conciseness.}. %
  The state space consists of three states for each variable $v \in V$: the state $v$ and the states $v\top$ and $v\bot$ that encode their assignment.
Additionally, we have a dedicated target $W$ and sink state $F$.   %
  We consider three actions: The actions \emph{true} ($\top$) and \emph{false} ($\bot$) semantically describe the assignment to existentially quantified variables. The action \emph{any} $\envact$ is used for all other states. Every environment reaches the target state iff one literal in the clause is assigned true.

  In the example, intuitively, a policy should assign the negation of $x$ to $y$. Formally, the policy $\sigma$, characterized by $ \sigma(\pi \cdot y) = \top$ iff $x_\bot \in \pi$, is winning.
\exampleend
\end{example}
As a consequence of this construction, we may also deduce the following theorem.%
\begin{theorem}
	Deciding whether a memoryless winning policy exists is NP-complete.
\end{theorem}
The proof of NP hardness uses a similar construction for the propositional SAT fragment of QBF, without universal quantifiers.
Additionally, the problem for memoryless policies is in NP, because one can nondeterministically guess a (polynomially sized) memoryless policy and verify in each environment independently.

\subsection{Policy Problem}
\label{sec:policy_problem}
Policies, mapping histories to actions, are generally infinite objects. 
However, we may extract winning policies from the BOMDP, which is (only) exponential in the MEMDP. 
Finite state controllers~\cite{DBLP:conf/uai/MeuleauPKK99} are a suitable and widespread representation of policies that require only a finite amount of memory. %
Intuitively, the number of memory states reflects the number of equivalence classes of histories that a policy can distinguish. %
In general, we cannot hope to find smaller policies than those obtained via a BOMDP.
\begin{restatable}{theorem}{exppolicy}\label{thm:exp_policy}
There is a family of MEMDPs $\{ \memdp^n \}_{n \geq 1}$ where for each $n$, $\memdp^n$ has $2n$ environments and $\mathcal{O}(n)$ states and where every winning policy for $\memdp^n$ requires at least $2^n$ memory states.
\end{restatable}
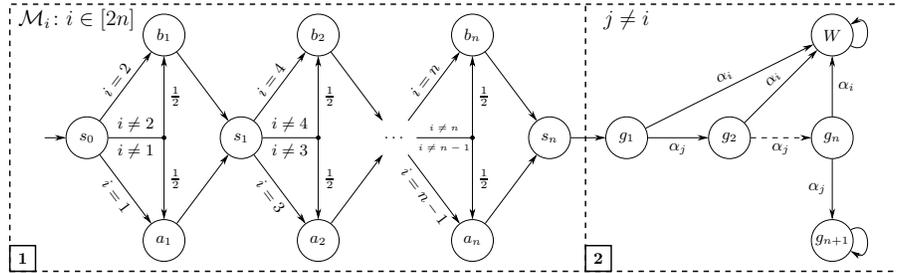
\begin{figure}[t]
  \centering
  \resizebox{\linewidth}{!}{
    \begin{tikzpicture}[
  >=stealth,
  initial text=$ $, %
  m/.style={ inner sep=0mm, outer sep=0mm },
  s/.style={ state, inner sep=0mm, minimum width=0.7cm },
  ]
  \node[s, initial] (s0) at (0,0) {$s_0$};
  \node[s] (a0) at (1.5,-2) {$a_1$};
  \node[s] (b0) at (1.5,2) {$b_1$};
  \node[m] (m0) at (1.5, 0) {};

  \draw[->] (s0) edge node[below, sloped] {$i=1$} (a0);
  \draw[->] (s0) edge node[above, sloped] {$i=2$} (b0);
  \draw (s0) -- node[below] {$i \ne 1$} node[above]{$i \ne 2$} (m0);
  \filldraw (m0) circle (1pt);
  \draw[->] (m0) edge node[right]{$\frac 12$} (a0);
  \draw[->] (m0) edge node[right]{$\frac 12$} (b0);

  \node[s] (s1) at (3,0) {$s_1$};
  \draw[->] (a0) edge (s1);
  \draw[->] (b0) edge (s1);

  \node[s] (a1) at (4.5,-2) {$a_2$};
  \node[s] (b1) at (4.5,2) {$b_2$};
  \draw[->] (s1) edge node[below,sloped] {$i=3$} (a1);
  \draw[->] (s1) edge node[above,sloped] {$i=4$}(b1);

  \node[m] (m1) at (4.5, 0) {};
  \draw (s1) --  node[below] {$i \ne 3$} node[above] {$i \ne 4$} (m1);
  \filldraw (m1) circle (1pt);
  \draw[->] (m1) edge node[right]{$\frac 12$} (a1);
  \draw[->] (m1) edge node[right]{$\frac 12$} (b1);

  \node[s, draw=none, fill=none] (ldots) at (6,0) {\ldots};
  \draw[->] (a1) edge (ldots);
  \draw[->] (b1) edge (ldots);

  \node[s] (sn2) at (9,0) {$s_n$};
  \node[s] (an2) at (7.5,-2) {$a_n$};
  \node[s] (bn2) at (7.5,2) {$b_n$};
  \node[m] (mn2) at (7.5, 0) {};

  \draw[->] (ldots) edge node[below, sloped] {$i=n-1$} (an2);
  \draw[->] (ldots) edge node[above, sloped] {$i=n$} (bn2);
  \filldraw (mn2) circle (1pt) ;
  \draw (ldots) -- node[below] {\tiny$i \ne n-1$} node[above] {\tiny$i \ne n$} (mn2) ;
  \draw[->] (mn2) edge node[right]{$\frac 12$} (an2);
  \draw[->] (mn2) edge node[right]{$\frac 12$} (bn2);

  \draw[->] (an2) edge (sn2);
  \draw[->] (bn2) edge (sn2);

  \node[s] (g1) at (10.5,0) {$g_1$};
  \draw[->] (sn2) edge (g1);
  \node[s] (W) at (14.5,2) {$W$};

  \node[s] (g2) at (12.5,0) {$g_2$};

  \node[s] (gn2) at (14.5,0) {$g_n$};

  \draw[->] (g1) edge node[above, sloped] {$\alpha_i$} (W) ;
  \draw[->] (g2) edge node[above, sloped] {$\alpha_i$} (W);
  \draw[->] (gn2) edge node[right] {$\alpha_i$} (W);

  \node[s] (F) at (14.5,-2) {$g_{n+1}$};
  \draw[->] (gn2) edge node[left] {$\alpha_j$} (F);
  \draw[->] (g1) edge node[below] {$\alpha_j$} (g2);
  \draw[->, dashed] (g2) edge node[below] {$\alpha_j$} (gn2);

  \node (M) at (-0.2, 2.3) {\large$\mathcal M_i\colon i \in [2n]$};
  \node (X) at (10.5, 2.3) {\large $j \neq i$};

  \draw[->] (W) to[out=-30, in=30, looseness=3] (W);
  \draw[->] (F) to[out=30, in=-30, looseness=3] (F);

  \draw[thick, dashed] (-1.5,-2.6) rectangle (16, 2.6);
  \draw[thick, dashed] (9.7, -2.6) -- (9.7, 2.6);

  \draw[thick] (-1.5,-2.6) rectangle node{\textbf{1}} (-1,-2.1);
  \draw[thick] (9.7,-2.6) rectangle node{\textbf{2}} (10.2,-2.1);
\end{tikzpicture}
  }
  \caption{Witness for exponential memory requirement for winning policies.}
  \label{fig:example_exp_policy}
\end{figure}
\noindent We illustrate the witness. 
Consider a family of MEMDPs $\{\memdp^n\}_n$, where $\memdp^n$ has $2n$ MDPs, $4n$ states partitioned into two parts, and at most $2n$ outgoing actions per state. We outline the MEMDP family in~\cref{fig:example_exp_policy}.
In the first part, there is only one action per state. The notation is as follows: in state $s_0$ and MDP $\memdp^n_1$, we transition with probability one to state $a_0$, whereas in $\memdp^n_2$ we transition with probability one to state $b_0$. 
In every other MDP, we transition with probability one half to either state. 
In state $s_1$, we do the analogous construction for environments $3$, $4$, and all others. 
A path $s_0b_1 \hdots$ is thus consistent with every MDP except $\memdp^n_1$.
The first part ends in state $s_n$.
By construction, there are $2^n$ paths ending in  $s_n$. 
Each of them is (in)consistent with a unique set of $n$ environments.
In the second part, a policy may guess $n$ times an environment  by selecting an action $\alpha_i$ for every $i \in [2n]$. 
Only in MDP $\memdp^n_i$, action $\alpha_i$ leads to a target state. 
In all other MDPs, the transition leads from state $g_j$ to $g_{j+1}$. The state $g_{n+1}$ is absorbing in all MDPs. 
Importantly, after taking an action $\alpha_i$ and arriving in $g_{j+1}$, there is (at most) one more MDP inconsistent with the path.

Every MEMDP $\memdp^n$ in this family has a winning policy which takes $\sigma(\pi \cdot g_i) = \alpha_{2i-1}$ if $a_{i} \in \pi$ and $\sigma(\pi \cdot g_i) = \alpha_{2i}$ otherwise. %
 Furthermore, when arriving in state $s_n$, the state of a finite memory controller must reflect the precise set of environments consistent with the history. 
There are $2^{n}$ such sets.
The proof shows that if we store less information, two paths will lead to the same memory state, but with different sets of environments being consistent with these paths. 
 As we can rule out only $n$ environments using the $n$ actions in the second part of the MEMDP, we cannot ensure winning in every environment.
\section{A Partial Game Exploration Algorithm}
\label{sec:algorithm}

In this section, we present an algorithm for the policy problem.
We tune the algorithm towards runtime instead of memory complexity, but aim to avoid running out of memory. 
We use several key ingredients to create a pragmatic variation of Alg.~\ref{alg:search_algorithm}, with support for extracting the winning policy.

First, we use an abstraction from BOMDPs to a belief stochastic game (BSG) similar to~\cite{DBLP:journals/tac/WintererJWJTKB21} that reduces the number of states and simplifies the iterative construction\footnote{At the time of writing, we were unaware of a polytime algorithm for BOMDPs.}.
Second, we tailor and generalize ideas from \emph{bounded model checking}~\cite{DBLP:journals/ac/BiereCCSZ03} to build and model check only a fragment of the BSG, using explicit \emph{partial exploration} approaches as in, e.g.,~\cite{DBLP:conf/icml/McMahanLG05,DBLP:conf/atva/BrazdilCCFKKPU14,DBLP:journals/tii/VolkJK18,DBLP:journals/lmcs/KretinskyM20}. 
Third, our exploration does not continuously extend the fragment, but can also prune this fragment by using the model checking results obtained so far. 
The structure of the BSG as captured by the environment graph makes the approach promising and yields some natural heuristics.
Fourth, the structure of the winning region allows to generalize results to unseen states. 
We thereby operationalize an idea from~\cite{DBLP:conf/cav/JungesJS21} in a partial exploration context. 
Finally, we analyze individual MDPs as an efficient and significant preprocessing step. 
In the following we discuss these ingredients.

\subsubsection{Abstraction to Belief Support Games.}
\label{sec:beliefsupportgames}
We briefly recap stochastic games (SGs). See~\cite{SGs,DBLP:conf/csl/ChatterjeeJH03} for more details.
\begin{definition}[SG]
  A \emph{stochastic game} is a tuple $\game = \tuple{\mdp, S_{\pa}, S_{\pb}}$, where $\mdp = \tuple{S, A, \initdist, p}$ is an MDP and $(S_{\pa}, S_{\pb})$ is a partition of $S$.
\end{definition}
$S_{\pa}$ are Player~1 states, and $S_{\pb}$ are Player~2 states.
As common, we also `partition' (memoryless deterministic) policies into two functions ${\sigma_{\pa} \colon S_{\pa} \rightarrow A}$ and ${\sigma_{\pa} \colon S_{\pb} \rightarrow A}$. 
A Player 1 policy $\sigma_{\pa}$ is winning for state $s$ if $\Pr(T \mid \sigma_\pa, \sigma_\pb)$ for all $\sigma_\pb$. We (re)use $\win{\bsg}{T}$ to denote the set of states with a winning policy.

We apply a game-based abstraction to group states that have the same observation. 
Player~1 states capture the observation in the BOMDP, i.e., tuples $\tuple{s, J}$ of MEMDP states $s$ and subsets $J$ of the environments. Player~1 selects the action $a$, the result is Player~2 state $\tuple{\tuple{s, J}, a}$. Then Player~2 chooses an environment $j\in J$, and the game mimics the outgoing transition from $\tuple{s, j, J}$, i.e., it mimics the transition from $s$ in $\memdp_j$. Formally:
\begin{definition}[BSG]
  Let $\bomdp$ be a BOMDP with $\bomdp = \tuple{\tuple{\stateset, \actionset, \initdist, p}, \observations, \obsfun}$.
  A \emph{belief support game} $\bsg$ for $\bomdp$ is an SG $\bsg = \tuple{\tuple{S', A', \initdist', p}, S_{\pa}, S_{\pb}}$ with $S' = S_{\pa} \cup S_{\pb}$ as usual, Player 1 states~$S_{\pa} = \observations$,  Player~2 states $S_\pb = Z \times A$, actions $A' = A \cup I$, initial distribution $\initdist'(\tuple{s, I}) = \sum_{i \in I} \initdist(\tuple{s, i, I})$, and the (partial) transition function $p$ defined separately for Player~1 and~2:
  \begin{align*}
    p'(\obs, a) &= \dirac({\tuple{\obs, a}}) &\text{(Player~1)}\\ %
    p'(\tuple{\obs, a}, j, \obs') &= p(\tuple{s, j, J}, a, \tuple{s', j, J'}) \text{ with }\obs = \tuple{s,J}, \obs' = \tuple{s',J'} &\text{(Player~2)}
  \end{align*}%
\end{definition}%
\begin{restatable}{lemma}{winningmemdpgame} \label{thm:winning_memdp_game}
An (acyclic) MEMDP $\memdp$ with target states $T$ is winning if(f) there exists a winning policy in the BSG $\bsg$ with target states $T_\observations$.\end{restatable}%
\noindent 
Thus, on acyclic MEMDPs, a BSG-based algorithm is sound and complete, however, on cyclic MDPs, it may not find the winning policy. The remainder of the algorithm is formulated on the BSG, we use sliced BSGs as the BSG of a sliced BOMDP, or equivalently, as a BSG with some states made absorbing.

\subsubsection{Main algorithm. }
\begin{algorithm}[t]
\begin{algorithmic}[1]
\Function{FindPolicy}{MEMDP $\memdp = \tuple{S, A, {\{p_i\}}_{i \in I}, \initdist}$, targets $T \subseteq S$}
    \State $\wintmp \assign \{\tuple{s, J} \mid s \in T, J \subseteq I\}$; $\notwin \assign \emptyset$; $i \assign 1$; $S_{\text{init}} \assign \supp(\initdist) \times \{I\}$
    \While{$S_{\text{init}} \cap \wintmp \neq \wintmp \text{ and } S_{\text{init}} \cap \notwin = \emptyset$}
        \State $\tuple{\game, F} \assign \text{GenerateGameSlice}(\memdp, \wintmp, \notwin, i)$
        \State $\wintmp \assign \wintmp \cup \win{\game}{W}$
        \State $\notwin \assign \notwin \cup S \setminus \win{\game}{ W \cup F}$
        \State $i \assign i + 1$
    \EndWhile

    \IfThenElse{$S_{\text{init}} \subseteq \wintmp$}%
    {\Return $\text{ExtractPolicy}(\wintmp)$}
    {\Return $\bot$}
\EndFunction
\end{algorithmic}
\caption{Policy finding algorithm}
\label{alg:find_policy}
\end{algorithm}
We outline~\cref{alg:find_policy} for the \emph{policy problem}.
We track the sets of almost-sure observations  and losing observations (states in the BSG).
Initially, target states are winning. 
Furthermore, via a simple preprocessing, we determine some winning and losing states on the individual MDPs.

We iterate until the initial state is winning or losing.
Our algorithm constructs a sliced BSG and decides \emph{on-the-fly} whether a state should be a frontier state, returning the sliced BSG and the used frontier states. 
We discuss the implementation below. 
 For the sliced BSG, we compute the winning region twice: Once assuming that the frontier states are winning, once assuming they are loosing.
 This yields an approximation of the winning and losing states, see~\cref{lem:winningfrontier}.
 From the winning states, we can extract a randomized winning policy~\cite{DBLP:journals/ai/ChatterjeeCGK16}.

\paragraph{Soundness.}
Assuming that the $\bsg$ is indeed a sliced BSG with frontier $F$. Then the following invariant holds:
$
 W \subseteq \win{\bsg}{T} \text{ and } L \cap \win{\bsg}{T} = \emptyset. 
 $
This invariant exploits that from a sliced BSG we can (implicitly) slice the complete BSG while preserving the winning status of every state, formalized below. 
In future iterations we only explore the implicitly sliced BSG. 
\begin{restatable}{lemma}{implicitslicing}\label{lem:implicitslicing}
	Given $W \subseteq \win{\bsg}{T_{\bsg}}$ and $L \subseteq S \setminus \win{\bsg}{T_{\bsg}}$:
  $
    \win{\bsg}{T_{\bsg}} = \win{\cut{\bsg}{W \cup L}}{T_{\bsg} \cup W}
  $
\end{restatable}

\noindent
\emph{Termination} depends on the sliced game generation. 
It suffices to ensure that in the long run, either $W$ or $L$ grow as there are only finitely many states. 
If $W$ and $L$ remain the same longer than some number of iterations, $W \cup L$ will be used as frontier.
Then, the new game will suffice to determine if $s \in W$ in one shot.

\begin{algorithm}[t]
\begin{algorithmic}[1]
\Function{GenerateGameSlice}{MEMDP $\memdp$, $W$, $L$, $i$}
    \State $Q \assign \{s_\iota \}$; $E = \{ s_\iota \}$
    \While{$s \in Q$ and $|E| \leq \texttt{Bound}[i]$ exists}
        \State $E \assign E \cup \{ s \}$ \Comment{Mark $s$ as explored}
        \State $\game \assign \cut{\bsg}{(S \setminus E)}$ \label{line:gameconstruction}\Comment{Extend game, cut-off everything not  explored}
        \State $Q \assign \reachable{\game} \setminus (E \cup W \cup L)$ \Comment{Add newly reached states}     
    \EndWhile
    \State \Return $\game, Q$
\EndFunction
\end{algorithmic}
\caption{Game generation algorithm}
\label{alg:game_generation}
\end{algorithm}

\subsubsection{Generating the sliced BSG.}
Algorithm~\ref{alg:game_generation} outlines the generation of the sliced BSG. In particular, we explore the implicit BSG from the initial state but make every state that we do not explicitly explore absorbing. In every iteration, we first check if there are states in $Q$ left to explore and if the number of explored states in $E$ is below a threshold $\textsf{Bound}[i]$.
Then, we take a state from the priority queue and add it to $E$. We find new reachable states\footnote{In l.~\ref{line:gameconstruction} we do not rebuild the game $\game$ from scratch but incrementally construct the data structures. Likewise, reachable states are a direct byproduct of this construction.} and add them to the queue~$Q$.

\subsubsection{Generalizing the winning and losing states. }
We aim to determine that a state in the game $\bsg$ is winning without ever exploring it. First, observe:
\begin{lemma}\label{lemma:winning_losing_inclusion}
	A winning policy in MEMDP $\memdp$ is winning in $\restrictenv{\memdp}{J}$ for any $J$.
\end{lemma}  
A direct consequence is the following statement for two environments $J_1 \subseteq J_2$:
\[ \tuple{s, J_{2}} \in \win{\bsg}{T} \quad\text{implies}\quad \tuple{s, J_{1}}  \in \win{\bsg}{T} . \]
Consequently, we can store $W$ (and symmetrically, $L$) as follows. For every MEMDP state $s \in S$, $W_s = \{ J \mid \tuple{s, J} \in W\}$ is downward closed on the partial order $P=(I, \subset)$. This allows for efficient storage: We only have to store the set of pairwise maximal elements, i.e., the antichain,
\[ W_s^{\max} = \{ J \in W_s \mid \forall J' \in W_s \text{ with } J \not\subseteq J' \}.\] 
To determine whether $\tuple{s,J}$ is winning, we check whether $J \subseteq J'$ for some $J' \in W_s^{\max}$. Adding $J$ to $W_s^{\max}$ requires removing all $J' \subseteq J$ and then adding $J$. Note, however, that $|W_s^{\max}|$ is still exponential in $|I|$ in the worst case.

\subsubsection{Selection of heuristics. }
The algorithm allows some degrees of freedom.
We evaluate the following aspects empirically.
\begin{inparaenum}[(1)]
\item 
The maximal size $\texttt{bound}[i]$ of a sliced BSG at iteration $i$ is critical. 
If it is too small, the sets $W$ and $L$ will grow slowly in every iteration. 
The trade-off is further complicated by the fact that the sets $W$ and $L$ may generalize to unseen states.
\item 
For a fixed $\texttt{bound}[i]$, it is unclear how to prioritize the exploration of states.
The PSPACE algorithm suggests that going deep is good, whereas the potential for generalization to unseen states is largest when going broad. 
\item 
Finally, there is overhead in computing both $W$ and $L$. 
If there is a winning policy, we only need to compute $W$. 
However, computing $L$ may ensure that we can prune parts of the state space. 
A similar observation holds for computing $W$ on unsatisfiable instances.
\end{inparaenum}

\begin{remark}
Algorithm~\ref{alg:find_policy} can be mildly tweaked to meet the PSPACE algorithm in Algorithm~\ref{alg:search_algorithm}.
The priority queue must ensure to always include complete (reachable) local BSGs and to explore states $\tuple{s, J}$ with small $J$ first. 
Furthermore, $W$ and $L$ require regular pruning, and we cannot extract a policy if we prune $W$ to a polynomial size bound. 
Practically, we may write pruned parts of $W$ to disk. 
\end{remark}

\section{Experiments}%
\label{sec:implementation_experiments}

We highlight two aspects: (1) A comparison of our prototype to existing baselines for POMDPs, and (2) an examination of the exploration heuristics. The technical report~\cite{technicalreport} contains details on the implementation, the benchmarks, and more results.

\paragraph{Implementation.}
We provide a novel \emph{PArtial Game Exploration} (\PaGE{}) prototype, based on \cref{alg:find_policy}, on top of the probabilistic model checker \storm~\cite{DBLP:journals/sttt/HenselJKQV22}.
 We represent MEMDPs using the \prism language with integer constants. 
 Every assignment to these constants induces an explicit MDP. 
SGs are constructed and solved using existing data structures and graph algorithms.

\paragraph{Setup.}
We create a set of benchmarks inspired by the POMDP and MEMDP literature~\cite{DBLP:conf/cav/JungesJS21,DBLP:conf/aaai/ChatterjeeCD16,DBLP:conf/tacas/HartmannsKPQR19}. 
We consider a combination of satisfiable and unsatisfiable benchmarks.
In the latter case, a winning policy does not exist.
We construct POMDPs from MEMDPs as in \cref{def:unionpomdp}. 
As baselines, we use the following two existing POMDP algorithms. 
For almost-sure properties, a \emph{belief-MDP construction}~\cite{DBLP:conf/atva/BorkJKQ20} acts similar to an efficiently engineered variant of our game-construction, but tailored towards more general quantitative properties. A \emph{SAT-based approach}~\cite{DBLP:conf/cav/JungesJS21} aims to find increasingly larger policies. 
We evaluate all benchmarks on a system with a 3GHz Intel Core i9-10980XE processor. 
We use a time limit of 30 minutes and a memory limit of 32 GB.

\paragraph{Results.}
\begin{figure}[t]
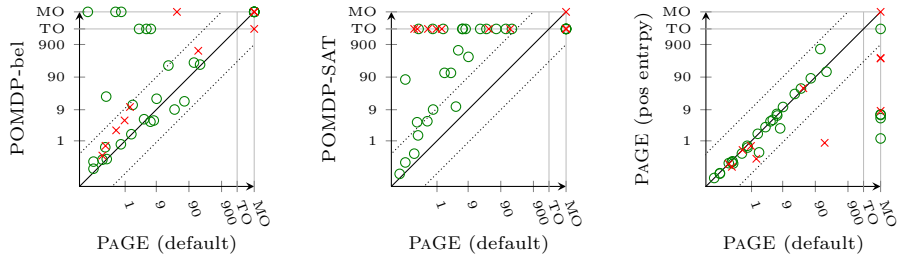

  \centering
  \makebox[\textwidth][c]{
  \subfigure {
    \benchmark{heuristicnentropylowerupper}{\PaGE{} (default)}{pomdpbelief}{POMDP-bel}
  }%
  \subfigure {
    \benchmark{heuristicnentropylowerupper}{\PaGE{} (default)}{pomdpsat}{POMDP-SAT}
  }%
  \subfigure {
    \benchmark{heuristicnentropylowerupper}{\PaGE{} (default)}{heuristicentropylowerupper}{\PaGE{} (pos entrpy)}
  }%
  }
  \vspace{-1.4em}
  \caption{Performance of baselines and novel \PaGE{} algorithm}\label{fig:scatter_plots}
  \vspace{-1.4em}
\end{figure}%
\cref{fig:scatter_plots} shows the (log scale) performance comparisons between different configurations\footnote{Every point $\tuple{x,y}$ in the graph reflects a benchmarks which was solved by the configuration on the x-axis in x time and by the configuration on the y-axis in y time. Points above the diagonal are thus faster for the configuration on the x-axis.}.
Green circles reflect satisfiable and red crosses unsatisfiable benchmarks. 
On the x-axis is \PaGE{} in its default configuration.
The first plot compares to the belief-MDP construction. 
The tailored heuristics and representation of the belief-support give a significant edge in almost all cases. 
The few points below the line are due to a higher exploration rate when building the state space. 
The second plot compares to the SAT-based approach, which is only suitable for finding policies, not for disproving their existence. 
This approach implicitly searches for a particular class of policies, whose structure is not appropriate for some MEMDPs.
The third plot compares \PaGE{} in the default configuration -- with negative entropy as priority function -- with \PaGE{} using positive entropy. As expected, different priorities have a significant impact on the performance. %

\newcommand{\mo}{MO}
\begin{table}[t]
  \centering
  \caption{Satisfiable and unsatisfiable benchmark results}
  \resizebox{0.52\linewidth}{!}{
  \begin{tabular}{lrrR{0.5cm}|R{1.0cm}R{1.15cm}R{1.0cm}R{1.0cm}|rr}
    \toprule
    &&&& \multicolumn{2}{c}{\PaGE (posentr)} & \multicolumn{2}{c}{\PaGE (negentr)} & \thead{Belief} & \thead{SAT}\\\ %
    & $|I|$ & $|S|$ & $|A|$ & \thead{t} & \thead{n} & \thead{t} & \thead{n} & \thead{t} & \thead{t}\\
    \midrule
    \multirow{3}{*}{\rotatebox{90}{Grid}} %
&19&132&4 &0.2 & 3002&0.2 & 3002&0.6 &3.7 \\
&39&152&4 &0.4 & 9007&1.6 & 41029&12.6 &121.3 \\
&199&474&4&6.4 & 337177&MO& &MO&TO\\
    \cmidrule(lr){1-10}
    \multirow{3}{*}{\rotatebox{90}{Catch}} %
&256&625&4  &6.6 & 93614&5.9 & 41094&3.8 &TO\\
&256&6561&4 &40.1 & 749295&32.6 & 337899&9.1 &TO\\
&256&14641&4&82.5 & 1826922&65.3 & 338079&16.2 &TO\\
    \cmidrule(lr){1-10}
    \multirow{3}{*}{\rotatebox{90}{Exp}} %
&8&19&9&0.1 & 349&0.1 & 349&0.1 &75.9 \\
&20&43&21&131.4 & 192163&197.6 & 448443&217.6 &TO\\
&24&51&25&TO& &MO& &MO&TO\\
    \cmidrule(lr){1-10}
    \multirow{4}{*}{\rotatebox{90}{Frogger}} %
&10&1200&4&0.2 & 1200&0.2 & 1200&22.7 &1.4 \\
&20&1200&4&0.4 & 1200&0.5 & 1200&MO&3.9 \\
&80&4000&4&4.4 & 4000&4.4 & 4000&TO&597.3 \\
&99&4000&4&5.9 & 8001&6.1 & 8001&TO&TO\\
    \bottomrule
  \end{tabular}
  }
  \resizebox{0.47\linewidth}{!}{
  \begin{tabular}{lrrr|R{1.0cm}R{1.0cm}R{1.0cm}R{1.0cm}|R{1.0cm}}
    \toprule
    &&&& \multicolumn{2}{c}{\PaGE (posentr)} & \multicolumn{2}{c}{\PaGE (negentr)} & \thead{Belief}\\
    &$|I|$&$|S|$&$|A|$& \thead{t} & \thead{n} & \thead{t} & \thead{n} & \thead{t} \\
    \midrule
    \multirow{4}{*}{\rotatebox{90}{MMind}} %
&16&21&16&0.1& 1003&0.2& 1445&0.3\\
&27&17&27&0.5& 5167&0.5& 7579&2.0\\
&32&25&32&0.6& 7799&0.9& 11809&4.2\\
&81&21&81&41.1& 170291&38.6& 296407&MO\\

    \cmidrule(lr){1-9}
    \multirow{3}{*}{\rotatebox{90}{Exp}} %
&20&42&21&0.8& 9005&173.8& 388127&576.1\\
&24&50&25&8.3& 41022&MO& &MO\\
&32&66&33&347.7& 337177&MO& &MO\\
    \bottomrule
  \end{tabular}
  }
  \label{table:benchmarks}
\end{table}%

\cref{table:benchmarks} shows an overview of satisfiable and unsatisfiable benchmarks. %
Each table shows the number of environments, states, and actions-per-state in the MEMDP. For \PaGE{}, we include both the default configuration (negative entropy) and variation  (positive entropy). For both configurations, we provide columns with the time and the maximum size of the BSG constructed.  We also include the time for the two baselines. Unsurprisingly, the number of states to be explored is a good predictor for the performance and the relative performance is as in Fig.~\ref{fig:scatter_plots}.%

\section{Conclusion}
\label{sec:conclusion}
This paper considers multi-environment MDPs with an arbitrary number of environments and an almost-sure reachability objective. 
We show novel and tight complexity bounds and use these insights to derive a new algorithm. 
This algorithm outperforms approaches for POMDPs on a broad set of benchmarks. 
For future work, we will apply an algorithm directly on the BOMDP~\cite{DBLP:journals/jcss/ChatterjeeCT16}. %

\clearpage
\pagebreak
\subsection*{Data-Availability Statement}
Supplementary material related to this paper is openly available on Zenodo at: \url{https://doi.org/10.5281/zenodo.7560675}
\bibliography{literature.bib}
\bibliographystyle{plain}
\clearpage
\pagebreak
\appendix
\section{Definition of FSCs for MEMDPs}
\label{app:fscs}

As a more concrete representation of a policy, we will use finite state controllers (FSCs). Due to \cref{lem:acyc_det_suffices} and the fact that the construction in \cref{thm:exp_policy} uses an acyclic MEMDP, we simplify matters and only consider \emph{deterministic FSCs}.

\begin{definition}[FSCs]
  A finite state controller (FSC) is a tuple $C = (Q, A, \delta, a, q_\iota)$ where:
  \begin{itemize}
    \item $Q$ is the finite state set, i.e., the policy's memory,
    \item $A$ is the finite action set,
    \item $\delta\colon Q \times S \to Q$ is the transition function,
    \item $a\colon Q \times S \to A$ is the action function,
    \item $q_{\iota} \in Q$ is the initial state.
  \end{itemize}
\end{definition}
A finite state controller is an automaton that can be used as policy for an MDP  (and MEMDP/BSG by extension).
The FSC and MDP supply each other with input in alternating fashion.
Initially, the FSC supplies the initial action $a(q_{\iota}, \initstate)$, after which the MDP performs this action and transitions to some state $s'$.
The FSC then transitions from memory state $q$ to $q'=\delta(q, s')$.
This process is then repeated indefinitely.
We extend the transition function to MDP paths:%
\begin{align*}
  \delta^{+}(\initstate) &= q_{\iota}\\
  \delta^{+}(\initstate \hdots s\ i\ s') &= \delta(\delta^{+}(\initstate \hdots s), s')
\end{align*}
Intuitively, $\delta^{+}(\pi)$ denotes the memory state reached in the finite state controller after processing $\pi$.
We can then define a policy $\sigma$ using the finite state controller:
\[
  \sigma(\pi) = a(\delta^{+}(\pi))
\]
The size of a policy is then the size of its state space $|Q|$.
A policy is \emph{memoryless} when $|Q| = 1$.

\clearpage
\section{Proofs regarding POMDPs in \cref{sec:pomdps}}
\subsection{Proof of \cref{thm:unionpomdp}}
Recall \cref{thm:unionpomdp}:
\unionpomdp*
\begin{proof}
  The observation set of $\union{\memdp}$ is equal to the state set of $\memdp$, $O=S$. Because of this, we can map a policy on $\memdp$ one-to-one to an (observation-based) policy on $\union{\memdp}$ by treating states as observations, and vice versa.
  \item($\Rightarrow$): Given a winning policy $\sigma$ for MEMDP $\memdp$, we show that $\sigma$ is also winning for union-POMDP $\union{\memdp}$. By definition of a winning policy for a MEMDP, this means that the probability for reaching the target state in each environment is equal to 1. By construction of $\union{\memdp}$, this also mean that the probability of reaching the target state is 1 for each of the initial states of $\union{\memdp}$, which implies that $\sigma$ is winning in $\union{\memdp}$.
  \item($\Leftarrow$): Given a winning policy $\sigma$ for union-POMDP $\union{\memdp}$, we show that $\sigma$ is also winning for MEMDP $\memdp$. We know that $\sigma$ is winning for the initial distribution $\initdist$, which implies that the probability of reaching the target set $T$ is 1 in each of the initial states $s_{i} \in \supp(\initdist)$. By construction, each initial state $s_{i}$ belongs to a disjoint part of the $\union{\memdp}$, which is the same as an environment in the MEMDP, and thus, $\sigma$ wins in all environments and consequently the MEMDP.
  \qed
\end{proof}

\clearpage
\section{Proofs regarding BOMDPs in \cref{sec:beliefsupportgames}}
\subsection{Proof of \cref{lem:paths}}
Recall \cref{lem:paths}:
\pathslemma*
\begin{proof}
  \item ($\Rightarrow$):
  The existence of the path $s_{1} \hdots s_{n}$ in $\memdp_{i}$ tells us that for all $k$: $p_{i}(s_{k}, a_{k})(s_{k+1}) > 0$.
  As stated by the lemma, we know that $i \in J_{1}$.
  Inductively applying the definition of $J_{k}$ we obtain that indeed for all $1 \le k \le n$: $i \in J_{k}$, and thus that $i \in J_{n}$.

  \item ($\Leftarrow$):
  We have that $J_{n} \subseteq J_{n-1} \subseteq \cdots \subseteq J_{1}$, due to the definition of the belief update function $\mathsf{Up}$.
  Combined with the assumption that $i \in J_{n}$ we have that for all $1 \le k \le n$: $i \in J_{k}$.
  Due to the definition of $J_{k}$, we obtain that $p_{i}(s_{k}, a_{k})(s_{k+1}) > 0$, which means that $s_{1} \hdots s_{n}$ is a path in~$\memdp_{i}$.
  \qed
\end{proof}

\subsection{Proof of \cref{lem:beliefobservation}}
Recall \cref{lem:beliefobservation}:
\beliefobservation*
\begin{proof}
  Assuming $O(\last(\pi))=O(\last(\pi'))$, we need to show that $\pathbelief(\pi)=\pathbelief(\pi')$.
  Without loss of generality, we fix $O(\last(\pi))=O(\last(\pi'))=\tuple{s, J}$ and $\pi = \tuple{s_{0}, i, J_{0}} a_{0} \tuple{s_{1}, i, J_{1}} \ldots \tuple{s_{k}, i, J_{k}}$ and $\pi' = \tuple{s_{0}', i', J'_{0}} a'_{0} \tuple{s'_{1}, i', J'_{1}} \ldots \tuple{s'_{l}, i', J'_{l}}$. This means that $J_{k}=J'_{l}=J$ and $s_{k}=s'_{l}=s$.

  Due to the fact that $\pi$ and $\pi'$ are paths in a BOMDP, we know that \[p(\tuple{s_{k}, i, J_{k}}, a_{k}, \tuple{s_{k+1}, i, J_{k+1}}) > 0\] for all $k$, which implies $J_{k+1} = \update{J_{k}}{s_{k}}{a_{k}}{s_{k+1}}$. In other words, $J$ is equal to repeatedly applying $\mathsf{Up}$, which is exactly the definition of $\pathbelief(\pi)$.
  \qed
\end{proof}

\subsection{Proof of \cref{thm:winning_memdp_bomdp}}\label{sec:proof_memdp_bomdp}
Recall \cref{thm:winning_memdp_bomdp}:
\winningmemdpbomdp*
\begin{proof}
  This follows from a more general result shown in~\cite[Theorem 20]{DBLP:conf/concur/SuilenVJ24}.
  We observe that their definition of BOMDPs is slightly different but equivalent, as the observation function reveals only the state, similar to MEMDPs.
  The theorem states that a policy that wins a Rabin objective for a given state $s$ for every environment $j \in J$ is also winning in the BOMDP state $\tuple{s, J}$, and vice-versa (applying necessary liftings).
  Without loss of generality, we can assume that each target state is a sink state, which means that reachability trivially reduces to a Rabin objective and we can apply~\cite[Theorem 20]{DBLP:conf/concur/SuilenVJ24} to our setting.
\qed
\end{proof}

\subsection{Proof for \cref{lem:winningfrontier}} \label{proof:winningfrontier}
Recall \cref{lem:winningfrontier}:
\winningfrontier*
\begin{proof}
  \item ($\win{\cut{\bomdp}{F}}{T} \subseteq \win{\bomdp}{T}$):
  Assuming that $s \in \win{\cut{\bomdp}{F}}{T}$, we have to show that $s \in \win{\bomdp}{T}$.
  Suppose policy $\sigma$ is winning in $\cut{\bomdp}{F}$, this means that the probability of reaching a state in $F$ is 0.
  Thus, all paths remain in the `sub-BOMDP' that is equal in $\bomdp$, which shows that $\sigma$ is also winning for $s$ in $\bomdp$ with target set $T$, i.e., $s \in \win{\bomdp}{T}$.

  \item ($\win{\bomdp}{T} \subseteq \win{\cut{\bomdp}{F}}{T \cup F}$):
  Assuming that $s \in \win{\bomdp}{T}$, we have to show that $s \in \win{\cut{\bomdp}{F}}{T \cup F}$.
  Suppose policy $\sigma$ is winning in $\bomdp$ with target set $T$.
  Any path produced by $\sigma$ almost-surely reaches a state in $T$.
  Along such a path we possibly visit a state in $F$, which would be winning in $\cut{\bomdp}{F}$ with target set $T \cup F$.
  $\cut{\bomdp}{F}$ is equal to $\bomdp$ except for states in $F$, so $\sigma$ is also winning in state $s$ of $\cut{\bomdp}{F}$ with target set $T \cup F$, i.e., $s \in \win{\cut{\bomdp}{F}}{T \cup F}$.
  \qed
\end{proof}

\subsection{Proof for \cref{lem:localpath}}
Recall \cref{lem:localpath}:
\localpath*
\begin{proof}
  Given a path $\pi = \pi_{1}\ldots\pi_{n}$ in $\bomdp$ s.t.\ every $\pi_{i}$ is $J_{i}$-local.
  We fix $\pi_{i}=\tuple{s_{i,0}, j, J_{i}}a_{i,0} \ldots \tuple{s_{i,k},j,J_{i}}$.
  We need to show that there is a path $J_{1} \ldots J_{n}$ in $\egraph$.
  Because every path fragment $\pi_{i}$ is $J_{i}$-local, we must have that $J_{i+1}=\update{J_{i}}{s_{i,k}}{a_{i,k}}{s_{i+1,0}}$, which implies the existence the transition from $J_{i}$ to $J_{i+1}$ in $\egraph$.
  \qed
\end{proof}

\subsection{Proof of \cref{lem:winninglocally}}\label{proof:winninglocally}
Recall \cref{lem:winninglocally}:
\winninglocally*%
\begin{proof}
  We show that for all states $s$
  \[
    \tuple{s, J} \in O\left(\win{\localbomdp{J}}{T'_{\bomdp}}\right)
    \Leftrightarrow
    \tuple{s, J} \in O\left(\win{\bomdp}{T_{\bomdp}}\right)
  \]
  which implies the lemma.

\item{($\Rightarrow$)}
Assuming $\tuple{s, J} \in O\left(\win{\localbomdp{J}}{T'_{\bomdp}}\right)$, we know there exists an observation-based policy $\sigma$ that is winning for observation $\tuple{s, J}$.
This means that playing $\sigma$ means that either a state $s'$ in $T_{\bomdp}$ or in $\win{\bomdp}{T_{\bomdp}} \setminus S_J$ is almost-surely reached.
In case a path reaches a state in $T_{\bomdp}$, the path would be winning in $\bomdp$ too.
Otherwise, such a path reaches a state $s'$ in $\win{\bomdp}{T_{\bomdp}} \setminus S_J$, which has some winning policy $\sigma'$ in $\bomdp$.
This shows us that we can construct a new policy for $\bomdp$ that wins in $s$, by switching to $\sigma'$ at the frontier state $s'$.

\item{($\Leftarrow$)}
Assuming $\tuple{s, J} \in O\left(\win{\bomdp}{T_{\bomdp}}\right)$, we know there exist an observation-based policy $\sigma$ that is winning for observation $\tuple{s, J}$.
This policy must also be winning $\localbomdp{J}$ with target set $T'_{\bomdp}$.
If a path produced by this policy stays within $S_J$, it must almost-surely reach a target state in $T_{\bomdp}$ (because it is winning).
If the path goes outside of $S_J$, it shows us that the frontier state $s' \in S \setminus S_J$ must be winning, i.e., $s' \in \win{\bomdp}{T_{\bomdp}} \setminus S_J$.
Thus, $\sigma$ almost-surely reaches target set $T'_{\bomdp}$ in $\localbomdp{J}$.
\qed
\end{proof}

\subsection{Proof of \cref{lem:locbsgsize}}\label{proof:locbsgsize}
Recall \cref{lem:locbsgsize}:
\locbsgsize*
\begin{proof}
  Recall the definition of a local BOMDP as (a fragment of) %
  \[ \localbomdp{J} = \onlyreachable{\cut{\custombomdp{\restrictenv{\memdp}{J}}}{F}}{\unif(S_J)}\quad\text{ where }\quad F = \observations \setminus \observations_J\ . \]
  There are $|S| \cdot |J|$ many states in $S_J$. %
  In each of these states, at most $|A|$ actions are enabled.
  For each state-action pair in the MEMDP, we may reach up to $|S|$ many states.
  In particular, for every $\tuple{s,j,J}$ and action $a$,
  it holds that if $\tuple{s',j, J'}$ and $\tuple{s',j, J''}$ can be reached in the local BOMDP from $\tuple{s,j,J}$, then $J' = J''$.
  There are thus $\mathcal{O}(|S|^2 \cdot |A| \cdot |J|)$ many states.
  However, the number of transitions is duplicated, and are thus $\mathcal{O}(|S|^2 \cdot |A| \cdot |J|^2)$ many.
  \qed
\end{proof}

\clearpage

\clearpage
\section{The PSPACE hardness proof (\cref{lem:decision_pspace_hard})}
\label{proof:pspacehard}

Recall~\cref{lem:decision_pspace_hard}:
\pspacehard*

\begin{figure}[H]
  \centering
  \begin{tikzpicture}
\node (x1) {$\exists x_{1}$};
\node (y1) [below=of x1] {$\forall y_{1}$};

\draw[->] (x1) -- node[right]{$a_{1}$} (y1);

\node (x2_0) [below=of y1, xshift=-1.5cm]{$\exists x_{2}$};
\node (x2_1) [below=of y1, xshift=1.5cm]{$\exists x_{2}$};

\draw[->] (y1) -- node[above, yshift=0.1cm]{$\top$} (x2_0);
\draw[->] (y1) -- node[above, yshift=0.1cm]{$\bot$} (x2_1);

\node (y2_0) [below=of x2_0]{$\forall y_{2}$};
\node (y2_1) [below=of x2_1]{$\forall y_{2}$};

\draw[->] (x2_0) -- node[right]{$a_{2}$} (y2_0);
\draw[->] (x2_1) -- node[right]{$a_{3}$} (y2_1);

\node (x3_0) [below=of y2_0, xshift=-1cm]{\vdots};
\node (x3_1) [below=of y2_0, xshift=1cm]{\vdots};
\node (x3_2) [below=of y2_1, xshift=-1cm]{\vdots};
\node (x3_3) [below=of y2_1, xshift=1cm]{\vdots};

\draw[->, dashed] (y2_0) -- node[above]{$\top$} (x3_0);
\draw[->, dashed] (y2_0) -- node[above]{$\bot$} (x3_1);
\draw[->, dashed] (y2_1) -- node[above]{$\top$} (x3_2);
\draw[->, dashed] (y2_1) -- node[above]{$\bot$} (x3_3);
\end{tikzpicture}
  \caption{Solution tree of $\Psi$}
  \label{fig:solution_set_qbf}
\end{figure}
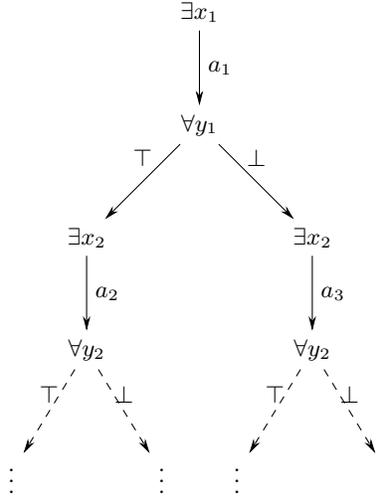

\paragraph{Solution trees for TQBF\@.}
Before we provide a proof, let us first remark that for every QBF formula $\Psi$ in the aforementioned form, we can express the truth value recursively as follows:
\begin{align*}
  \Psi &=
  E_{1}([x_{1} \mapsto \top]) \vee E_{1}([x_{1} \mapsto \bot]) \text{ where}\\
  E_{1}(m) &= F_{1}(m \cdot [y_{1} \mapsto \top]) \wedge F_{1}(m \cdot [y_{1} \mapsto \bot])\text{ where}\\
  F_{1}(m) &= E_{2}(m \cdot [x_{2} \mapsto \top]) \vee E_{2}(m \cdot [x_{2} \mapsto \bot])\text{ where}\\
  \vdots\\
  F_{n}(m) &= \Phi[m]
\end{align*}
where $\Phi[m]$ expresses the truth value of $\Phi$ after applying the substitution $m$. This shows that a witness to a formula belonging to TQBF, such as $m$, is tree-shaped. \cref{fig:solution_set_qbf} shows an example of such a witness, which we will call a \emph{solution tree}.
Such a tree has the following properties:
\begin{itemize}
  \item Any path in the tree assigns exactly one value $\top$ or $\bot$ to each variable in the order as indicated by $\Psi$.
  \item Each $\forall$ variable must permit both assignments $\top$ and $\bot$.
  \item Each $\exists$ variable must permit one of the assignments $\top$ or $\bot$.
  \item The assignments $m$ belonging to any path satisfies $\Phi$.
\end{itemize}

\paragraph{Construction}
Given a QBF formula
\[
  \Psi = \exists x_{1} \forall y_{1} \exists x_{2} \forall y_{2} \hdots \exists x_{n} \forall y_{n} \Big[ \Phi \Big]
\]
over variables $V=\{x_1, y_1, x_2, y_2, \hdots, x_n, y_n\}$ where $\Phi$ is a Boolean formula over $V$.
We assume that $\Phi$ is in conjunctive normal form, i.e., it is a conjunction of disjunctions of positive and negative literals.
The $i$-th clause of the conjunction is denoted as $\Phi_{i}$.
\begin{figure}[t]
  \centering
  \resizebox{\linewidth}{!}{
    \begin{tikzpicture}
  \node[state] (x1) {$x_{1}$};
  \node[state] (x1t) [right=of x1, yshift=1cm] {$x_{1}\top$};
  \node[state] (x1f) [right=of x1, yshift=-1cm] {$x_{1}\bot$};

  \draw[->] (x1) -- node[above]{$\top$} (x1t);
  \draw[->] (x1) -- node[below]{$\bot$} (x1f);

  \node[state] (y1) [right=of x1t, yshift=-1cm, xshift=1cm] {$y_{1}$};

  \node (x1tm) [inner sep=0mm, right=of x1t] {};
  \node (x1fm) [inner sep=0mm, right=of x1f] {};

  \node[state] (x1tw) [above=of x1tm] {$W$};
  \node[state] (x1fw) [below=of x1fm] {$W$};

  \draw (x1t) -- node[above]{$\envact$} (x1tm);
  \draw (x1f) -- node[above]{$\envact$} (x1fm);

  \draw[->] (x1tm) -- node[right]{$[\Phi_{i}(x_{1})]$} (x1tw);
  \draw[->] (x1fm) -- node[right]{$[\Phi_{i}(\neg x_{1})]$} (x1fw);

  \draw[->] (x1tm) -- node[right,yshift= 0.1cm]{$[\neg\Phi_{i}(x_{1})]$} (y1);
  \draw[->] (x1fm) -- node[right,yshift=-0.1cm]{$[\neg\Phi_{i}(\neg x_{1})]$} (y1);

  \filldraw (x1tm) circle (1pt);
  \filldraw (x1fm) circle (1pt);

  \node (y1m) [inner sep=0mm, right=of y1] {};
  \filldraw (y1m) circle (1pt);
  \draw (y1) -- node[above]{$\envact$} (y1m);

  \node[state] (y1t) [right=of y1m, yshift=1cm] {$y_{1}\top$};
  \node[state] (y1f) [right=of y1m, yshift=-1cm] {$y_{1}\bot$};

  \draw[->] (y1m) -- node[above]{$\frac12$} (y1t);
  \draw[->] (y1m) -- node[below]{$\frac12$} (y1f);

  \node (y1tm) [inner sep=0mm, right=of y1t] {};
  \node (y1fm) [inner sep=0mm, right=of y1f] {};

  \filldraw (y1tm) circle (1pt);
  \filldraw (y1fm) circle (1pt);

  \draw (y1t) -- node[above]{$\envact$} (y1tm);
  \draw (y1f) -- node[above]{$\envact$} (y1fm);

  \node[state] (y1tw) [above=of y1tm] {$W$};
  \node[state] (y1fw) [below=of y1fm] {$W$};

  \draw[->] (y1tm) -- node[right] {$[\Phi_{i}(y_{1})]$}(y1tw);
  \draw[->] (y1fm) -- node[right] {$[\Phi_{i}(\neg y_{1})]$}(y1fw);

  \node[state] (x2) [right=of y1tm, yshift=-1cm] {$x_{2}$};
  \draw[->] (y1tm) -- node[right,xshift=-0.2cm,yshift=0.2cm] {$[\neg\Phi_{i}(y_{1})]$} (x2);
  \draw[->] (y1fm) -- node[right,xshift=-0.2cm,yshift=-0.2cm] {$[\neg\Phi_{i}(\neg y_{1})]$} (x2);

  \node[state] (xn1) [right=of x2] {$x_{n+1}$};

  \draw[->, dashed] (x2) -- (xn1);
  \draw[->] (xn1) to[in=45,out=-45,looseness=3] (xn1);

  \node (Q) [left=of x1, xshift=0.5cm] {$\mathcal{Q}_{i}\colon$};
  \draw [->] (Q) -- (x1);
\end{tikzpicture}
  }
  \caption{Construction of $\mathcal{Q}$. We use Iverson brackets, $[\Phi] = \begin{cases}
                                                                  1 & \text{if } \Phi\\
                                                                  0 & \text{otherwise}
                                                                \end{cases}$}
  \label{fig:construction_qbf}
\end{figure}
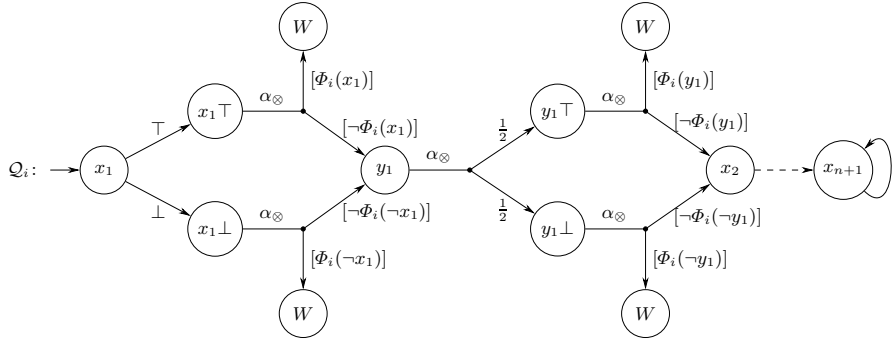

\newcommand{\nextvar}[1]{\mathsf{nextvar}(#1)}
To decide whether $\Phi$ is satisfiable, we construct a MEMDP which has one environment for each clause:
More precisely, the MEMDP $\mathcal{Q} = (S, A, \initstate, \{p_{i}\}_{i \in [n]})$ as illustrated by \cref{fig:construction_qbf}, has state space:
\[ S = \{W, x_{n+1}\} \cup  \left(\bigcup_{i \in [n]} \{x_{i}, x_{i}\top, x_{i}\bot, y_{i}, y_{i}\top, y_{i}\bot\}\right), \] actions $A = \{\top, \bot, \envact\}$ and initial state $\initstate = x_{1}$ as explained in the main text.
We define $\{p_{i}\}_{i \in [n]}$ in parts, for all $i, j \in [n], v \in \{x_{j}, y_{j}\}$:
\begin{itemize}
  \item $p_{i}(x_{j}, \top) = [x_{j}\top \mapsto 1]$,
  \item $p_{i}(x_{j}, \bot) = [x_{j}\bot \mapsto 1]$,
  \item $p_{i}(y_{j}, \envact) = [y_{j}\top \mapsto \sfrac12, y_{j}\bot \mapsto \sfrac12]$,

  \item $p_{i}(v\top, \envact) = \Big[ W \mapsto [\Phi_{i}(v)], \nextvar{v} \mapsto [\neg\Phi_{i}(v)] \Big]$,
  \item $p_{i}(v\bot, \envact) = \Big[ W \mapsto [\Phi_{i}(\neg v)], \nextvar{v} \mapsto [\neg\Phi_{i}(\neg v)] \Big]$,
  \item \ldots otherwise, $p_{i}(s, a) = [x_{n+1} \mapsto 1]$,
\end{itemize}
where $\nextvar{x_j} = y_j$ and $\nextvar{y_j} = x_{j+1}$.

The reduction can be performed in polynomial time.

We fix the target set to $\{W\}$.
We will now prove the following statement, which is equivalent to the original theorem:
$\mathcal{Q}$ has a winning policy $\Leftrightarrow \Psi$ is a true formula.

\paragraph{$(\Rightarrow)$}
Assuming we can obtain a winning deterministic policy $\sigma\colon (S \times A)^{*} \times S \to A$, we have to prove that $\Psi$ is a true formula.
We now show that $\sigma$ can be converted into a witness solution tree.
By the definition of $\mathcal{Q}$ and the definition of a winning policy, $\sigma$ and the $\mathcal{Q}$ have the following properties:
\begin{itemize}
  \item As $\sigma$ is winning, each path produced by $\sigma$ eventually reaches $W$, and thus satisfies $\Phi$ with the chosen assignment.
  \item Along every path, we reach either $v\top$ or $v\bot$ for every variable $v$. 
  \item $\sigma$ cannot prevent $y\bot$ or $y\top$ as the $y$ states only have a $\envact$ action, which chooses either assignment with equal probability.
\end{itemize}
In other words, the set of all paths produced by $\sigma$ forms a solution tree, after performing some necessary conversions.
Representing the solution tree as a set of assignments $m'_{n}$:
\begin{align*}
  m_{1} &= \begin{cases}
            \{ [x_{1} \mapsto \top]\} & \text{if } \sigma(x_{1}) = \top\\
            \{ [x_{1} \mapsto \bot]\} & \text{if } \sigma(x_{1}) = \bot\\
          \end{cases}\\
  m'_{1} &= \{ m \cdot [y_{1} \mapsto \top \mid m \in m_{1}]\} \cup \{ m \cdot [y_{1} \mapsto \bot] \mid m \in m_{1}\}\\
  m_{i+1} &= \begin{cases}
            \{ m \cdot [x_{i+1} \mapsto \top] \mid m \in m'_{i}\} & \text{if } \sigma(x_{i+1}) = \top\\
            \{ m \cdot [x_{i+1} \mapsto \bot] \mid m \in m'_{i}\} & \text{if } \sigma(x_{i+1}) = \bot\\
          \end{cases}\\
  m'_{i+1} &= \{ m \cdot [y_{i+1} \mapsto \top \mid m \in m_{i+1}]\} \cup \{ m \cdot [y_{i+1} \mapsto \bot] \mid m \in m_{i+1}\}\\
\end{align*}

\paragraph{$(\Leftarrow)$}
Assuming that $\Psi$ is a true formula, we know there must exist some witness solution tree $m$.
We can transform this solution tree into a winning policy $\sigma$ as follows:
\begin{align*}
  \sigma(x_{1}^{a_{1}} y_{1}^{a_{2}}\ldots x_{i}) &=
  \begin{cases}
    \top & \text{if } \exists p\colon [x_{1} \mapsto a_{1}, y_{1} \mapsto a_{2}, \ldots, x_{i} \mapsto \top] \cdot p \in m\\
    \bot & \text{if } \exists p\colon [x_{1} \mapsto a_{1}, y_{1} \mapsto a_{2}, \ldots, x_{i} \mapsto \bot] \cdot p \in m\\
  \end{cases}\\
  &\text{otherwise}\\
  \sigma(\pi) &= \envact
\end{align*}
This is winning because by the properties of the solution tree, any path generated by this satisfies all clauses in $\Phi$, which in turn means winning in $Q$.
\clearpage
\section{Exponential Policies (Proof of \cref{thm:exp_policy})}
\label{proof:exp_policy}
\begin{figure}[t]
  \centering
  \resizebox{\linewidth}{!}{
    \begin{tikzpicture}[
  >=stealth,
  initial text=$ $, %
  m/.style={ inner sep=0mm, outer sep=0mm },
  s/.style={ state, inner sep=0mm, minimum width=0.7cm },
  ]
  \node[s, initial] (s0) at (0,0) {$s_0$};
  \node[s] (a0) at (1.5,-2) {$a_1$};
  \node[s] (b0) at (1.5,2) {$b_1$};
  \node[m] (m0) at (1.5, 0) {};

  \draw[->] (s0) edge node[below, sloped] {$i=1$} (a0);
  \draw[->] (s0) edge node[above, sloped] {$i=2$} (b0);
  \draw (s0) -- node[below] {$i \ne 1$} node[above]{$i \ne 2$} (m0);
  \filldraw (m0) circle (1pt);
  \draw[->] (m0) edge node[right]{$\frac 12$} (a0);
  \draw[->] (m0) edge node[right]{$\frac 12$} (b0);

  \node[s] (s1) at (3,0) {$s_1$};
  \draw[->] (a0) edge (s1);
  \draw[->] (b0) edge (s1);

  \node[s] (a1) at (4.5,-2) {$a_2$};
  \node[s] (b1) at (4.5,2) {$b_2$};
  \draw[->] (s1) edge node[below,sloped] {$i=3$} (a1);
  \draw[->] (s1) edge node[above,sloped] {$i=4$}(b1);

  \node[m] (m1) at (4.5, 0) {};
  \draw (s1) --  node[below] {$i \ne 3$} node[above] {$i \ne 4$} (m1);
  \filldraw (m1) circle (1pt);
  \draw[->] (m1) edge node[right]{$\frac 12$} (a1);
  \draw[->] (m1) edge node[right]{$\frac 12$} (b1);

  \node[s, draw=none, fill=none] (ldots) at (6,0) {\ldots};
  \draw[->] (a1) edge (ldots);
  \draw[->] (b1) edge (ldots);

  \node[s] (sn2) at (9,0) {$s_n$};
  \node[s] (an2) at (7.5,-2) {$a_n$};
  \node[s] (bn2) at (7.5,2) {$b_n$};
  \node[m] (mn2) at (7.5, 0) {};

  \draw[->] (ldots) edge node[below, sloped] {$i=n-1$} (an2);
  \draw[->] (ldots) edge node[above, sloped] {$i=n$} (bn2);
  \filldraw (mn2) circle (1pt) ;
  \draw (ldots) -- node[below] {\tiny$i \ne n-1$} node[above] {\tiny$i \ne n$} (mn2) ;
  \draw[->] (mn2) edge node[right]{$\frac 12$} (an2);
  \draw[->] (mn2) edge node[right]{$\frac 12$} (bn2);

  \draw[->] (an2) edge (sn2);
  \draw[->] (bn2) edge (sn2);

  \node[s] (g1) at (10.5,0) {$g_1$};
  \draw[->] (sn2) edge (g1);
  \node[s] (W) at (14.5,2) {$W$};

  \node[s] (g2) at (12.5,0) {$g_2$};

  \node[s] (gn2) at (14.5,0) {$g_n$};

  \draw[->] (g1) edge node[above, sloped] {$\alpha_i$} (W) ;
  \draw[->] (g2) edge node[above, sloped] {$\alpha_i$} (W);
  \draw[->] (gn2) edge node[right] {$\alpha_i$} (W);

  \node[s] (F) at (14.5,-2) {$g_{n+1}$};
  \draw[->] (gn2) edge node[left] {$\alpha_j$} (F);
  \draw[->] (g1) edge node[below] {$\alpha_j$} (g2);
  \draw[->, dashed] (g2) edge node[below] {$\alpha_j$} (gn2);

  \node (M) at (-0.2, 2.3) {\large$\mathcal M_i\colon i \in [2n]$};
  \node (X) at (10.5, 2.3) {\large $j \neq i$};

  \draw[->] (W) to[out=-30, in=30, looseness=3] (W);
  \draw[->] (F) to[out=30, in=-30, looseness=3] (F);

  \draw[thick, dashed] (-1.5,-2.6) rectangle (16, 2.6);
  \draw[thick, dashed] (9.7, -2.6) -- (9.7, 2.6);

  \draw[thick] (-1.5,-2.6) rectangle node{\textbf{1}} (-1,-2.1);
  \draw[thick] (9.7,-2.6) rectangle node{\textbf{2}} (10.2,-2.1);
\end{tikzpicture}
  }
  \caption{Witness for exponential memory requirement for winning policies (Repeats \cref{fig:example_exp_policy} for convenience)}
  \label{fig:example_exp_policy_repeated}
\end{figure}
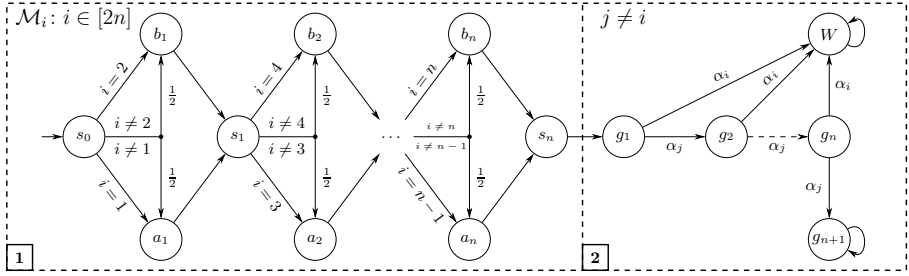

Recall \cref{thm:exp_policy}:
\exppolicy*

\begin{proof}
  We prove this by using the MEMDP family presented in \cref{fig:example_exp_policy_repeated} as a witness.
  We define a member of this family $\memdp^{n} = (S, A, \initstate, \{p_{i}\}_{i \in [2n]})$ as follows:
  \begin{itemize}
    \item $S = \{s_{0}, W, g_{n+1}\} \cup \bigcup_{i \in [n]} \{s_{i}, a_{i}, b_{i}\} \cup \{ g_{i} \mid i \in [n]\}$
    \item $A = \{\envact\} \cup \{\alpha_{i} \mid i \in [n]\}$
    \item $\initstate = s_{0}$
    \item We define $\{p_{i}\}_{i \in [2n]}$ in parts, for all $i, k \in [2n], j \in [n]$:
          \begin{itemize}
            \item $p_{i}(s_{j}, \envact) =
                  \begin{cases}
                    [a_{j+1} \mapsto 1] & \text{if } i=2j + 1\\
                    [b_{j+1} \mapsto 1] & \text{if } i=2j + 2\\
                    [a_{j+1} \mapsto \sfrac12, b_{j+1} \mapsto \sfrac12] & \text{otherwise}
                  \end{cases}$
            \item for all $s \in \{a_{j}, b_{j}\}\colon p_{i}(s, \envact) = [s_{j} \mapsto 1]$
            \item $p_{i}(g_{j}, \alpha_{k}) = \begin{cases}
 	[W \mapsto 1] & \text{ if i=k } \\
 	[g_{j+1} \mapsto 1] & \text{otherwise }
 \end{cases}$
          \end{itemize}
  \end{itemize}

  We show that any winning policy grows exponentially in $n$.

  Looking at the definition of $p_{i}(s_{j}, \envact)$, we see that for each $s_{j}$, most environments transition to $a_{j+1}$ and $b_{j+1}$ with equal probability.
  Exactly one environment only allows transitioning to state $a_{j+1}$ ($i=2j+1$), and exactly one environment only allows transitioning to state $b_{j+1}$ ($i=2j+2$).
  As a consequence, for path $\pi$ in environment $i$, we have:
  \begin{equation}\label{eq:oneenv}
    \Big[a_{j} \in \pi \Rightarrow \pi \not\in \Path_{\memdp_{2j+2}}\Big] \wedge \Big[b_{j} \in \pi \Rightarrow \pi \not\in \Path_{\memdp_{2j+1}}\Big]
  \end{equation}
  For any path $\pi$ such that $\last(\pi) = s_{n}$ we obtain a set of possible environments $B_{\pi}$ (notice that there is no nondeterminism in this part of the MEMDP).
  \begin{align*}
    B_{\pi} = \{ i \in [2n] \mid \forall \sigma.~ {\Pr}_{\memdp^{n}_{i}}(\pi \mid \sigma) > 0\}
  \end{align*}
 
  Due to \cref{eq:oneenv}, we have that $|B_{\pi}|=n$ and that the mapping from $\pi$ to $B_{\pi}$ is one-to-one.
  As a result, because there are $2^{n}$ possible paths leading to $s_{n}$, there are also $2^{n}$ possible sets of environments for $s_{n}$.

  The last part of the model requires `guessing' the correct environment.
  The only winning action in environment $i$ is $\alpha_{i}$.
  If the correct guess is not performed within $n$ guesses, a failing deadlock state is reached.

  To guarantee winning, the actions performed by the policy after observing $\pi$ must therefore be some permutation of $\{\alpha_{i} \mid i \in B_{\pi}\}$.
Now suppose for contradiction that there exists some policy $\sigma$, represented by an FSC $C=(Q, A, \delta, a, q_{\iota})$, which has $|Q| < 2^{n}$.
  Then, because there are $2^{n}$ paths reaching state $s_{n}$, there must be two paths $\pi \ne \pi', \last(\pi)=\last(\pi')=s_{n}$ such that $\delta^{+}(\pi) = \delta^{+}(\pi')$.
  It then follows that $\sigma(\pi \cdot \tau) = \sigma(\pi' \cdot \tau)$ for all $\tau \in (A \times S)^{*}$, due to the definition of an FSC policy.
  Because $\pi \ne \pi'$, we know that $B_{\pi} \ne B_{\pi'}$.
  Combined with the fact that $|B_{\pi}|=|B_{\pi'}|=n$, we know that $B_{\pi}$ and $B_{\pi'}$ cannot be permutations of each other, and the assumption that $\sigma$ is winning is false.

  \qed
\end{proof}

\clearpage

\section{Proofs for \cref{sec:algorithm}}
\subsection{Proof of \cref{thm:winning_memdp_game}}
Recall \cref{thm:winning_memdp_game}:
\winningmemdpgame*
\begin{proof}
  \paragraph{($\Rightarrow$):}
  Assuming that $\memdp$ is acyclic, we have that $\bsg$ is also acyclic.
  If we have a policy $\sigma$ that is winning in $\memdp$, we show that it must also be winning in $\bsg$.
  $\sigma$ winning means that it is winning in every environment $\memdp_{i}$.
  Due to acyclicity of $\memdp$, we know that $\sigma$ reaches a target state in a finite amount of transitions.
  Now assume for contradiction that $\sigma$ is not winning for $\bsg$, then we have some finite trace $\pi = s_{0,1}\ a_{0}\ s_{0,2}\ i_{0}\ s_{1,1}\ \ldots\ i_{n}\ s_{n,1}$ where $s_{n,1} \not\in T$.
  Due to monotonicity of the belief-support, we know that environment $i_{n}$ chosen by Player~2 could have been picked for each Player~2 state, i.e., $i_{k} = i_{n}$ for all $k$ in $\pi$.
  This describes a path which by definition of $\bsg$ is also present in $\memdp_{i}$, which contradicts the assumption that $\sigma$ is winning.

  \paragraph{($\Leftarrow$):}
  Assuming that there exists a winning policy $\sigma$ for $\bsg$, we can turn $\sigma$ into a policy $\sigma'$ which applies to $\memdp$ by projecting away the Player~2 states.
  If we consider Player~2 policies in the $\bsg$ that only pick 1 environment, the induced MDP is exactly that environment.
  We know that $\sigma$ wins against every Player~2 policy, and therefore it is winning in every environment.
  \qed
\end{proof}

\subsection{Proof for \cref{lem:implicitslicing}}
Recall \cref{lem:implicitslicing}:
\implicitslicing*
\begin{proof}%
  The theorem is equivalent to: $s \in \win{\bsg}{T_{\bsg}} \Leftrightarrow s \in \win{\cut{\bsg}{W \cup L}}{T_{\bsg} \cup W}$.
  We split the proof into the two parts.
  \item ($\Rightarrow$):
  Assuming $s \in \win{\bsg}{T_{\bsg}}$, we show that $s \in \win{\cut{\bsg}{W \cup L}}{T_{\bsg} \cup W}$.
  Due to our assumption, there exists some winning policy $\sigma$ for $s$ in $\bsg$.
  Then, necessarily, $\sigma$ must be winning in $\cut{\bsg}{W \cup L}$ with target set $T_{\bsg} \cup W$.
  All paths produced by $\sigma$ are identical in $\cut{\bsg}{W \cup L}$ except for those that contain a state in $W$ or $L$.
  In case a path reaches a $W$ state then it is clearly winning.
  Now suppose for contradiction that a path can contain a state in $L$.
  This would mean that $\sigma$ is capable of leaving the winning region, which contradicts our assumption that $\sigma$ is winning.

  \item ($\Leftarrow$):
  Assuming $s \in \win{\cut{\bsg}{W \cup L}}{T_{\bsg} \cup W}$, we show that $s \in \win{\bsg}{T_{\bsg}}$.
  We employ a similar argument to the proof given for \cref{lem:winninglocally}.
  Due to our assumption, there exists some winning policy $\sigma$ in $\cut{\bsg}{W \cup L}$ with target set $T_{\bsg} \cup W$.
  We can construct a new policy that is winning in $\bsg$ by extending $\sigma$.
  We follow $\sigma$, and if we reach a state $s' \in W$, we can switch to a policy that is winning for $s'$ in $\bsg$.
  \qed
\end{proof}

\clearpage

\section{Benchmark Data}
\subsection{Models} \label{sec:benchmark_models}
We will give a small description of each of the MEMDP families used in the benchmark.

\paragraph{Exponential.}
The same MEMDP family as the one used in the proof of \cref{thm:exp_policy}.
For MEMDP $\memdp_{n}$ the amount of states is about $4 \cdot n$ and contains $2 \cdot n$ environments.

\paragraph{Frogger.}
A grid of size $N \times M$ in which the agent starts at $(1,1)$ and tries to reach $(N, M)$ by moving in the cardinal directions.
Along this path, there is one car driving back and forth.
The amount of possible starting locations of this car is indicated by the amount of environments.
The amount of states is about $2 \cdot N^{2} \cdot M$.

\paragraph{Pacman.}
A grid of size $N \times M$ in which the agent starts at $(1,1)$ and tries to reach $(N, M)$ by moving in the cardinal directions.
Starting at $(1, M)$ there is a ghost that walks on the grid.
For each direction the agent moves, the ghost also moves in a fixed (but initially unknown) direction.
For each of the four cardinal directions the agent moves, the ghost can respond with any of the four cardinal directions, which is encoded by the environment.
Because of this, the MEMDP has a maximum of $4^{4}=256$ possible environments.

\paragraph{Catchman.}
Same as the Pacman family, but instead the ghost is the target.

\paragraph{Mastermind.}
This MEMDP family is a simplified version of the Mastermind board game.
The agent tries to guess a code consisting of a sequence of colored balls.
Specifically, the agent tries to guess a sequence of length $N$, in which the balls can be any of $C$ colors, within $G$ guesses.
After each guess the agent is able to observe the amount of balls that were exactly right (same color and same position).
Each environment represents one of the ($C^{N}$) possible codes.
If the correct code is guessed, a target state is reached.
When the last guess is not correct, a fail state is reached.

\paragraph{Infinite Belief.}
Dummy MEMDP designed to show that quantitative methods can `trip up'.
This is due to the fact that the POMDP constructed from this MEMDP has an infinite belief space representation.

\paragraph{Grid.}
A grid of size $N \times M$ in which the agent starts at $(1,1)$ and tries to reach $(N,M)$ by moving in the cardinal directions.
In this grid there is one hole, although its position is not certain.
If the agent walks into the hole, there is no way of getting out, thus failing to reach the target state.
Each environment represents a different hole location.
The agent can sense danger when one tile away from the hole.

\paragraph{Unsatisfiable variants.}
Some of the above (satisfiable) MEMDPs have a unsatisfiable counterpart. The Frogger MEMDPs have a wider car, which makes it impossible to cross the street. The Pacman MEMDPs have specific dimensions such that the ghost that mirrors the agents movements is always capable of preventing the agent from reaching the target state. The Mastermind and Exponential MEMDPs allow one guess fewer than is necessary to obtain a winning policy.

\subsection{Heuristics} \label{sec:benchmark_heuristics}
Here we discuss all the heuristics that were tested in the benchmark.
For every heuristic, we tested three different variants.
These variants are related to whether in the \PaGE{} algorithm we compute the lower bound, the upper bound, or both.
In \cref{table:benchsat1,table:benchsat2,table:benchunsat1,table:benchunsat2} we refer to these variants by -, + and -+ respectively. \cref{fig:extra_scatter_plots1,fig:extra_scatter_plots2} illustrate the effects of using different heuristics.

\paragraph{BFS\@. }
Breadth-first search heuristic which explores game states in the order they are first discovered.
\paragraph{DFS\@. }
Depth-first search heuristic which explores game states in the order they are last discovered.
\paragraph{Entropy\@. }%
Positive entropy heuristic which explores game states with the smallest environment set first.
\paragraph{NEntropy\@. }%
Negative entropy heuristic which explores game states with the biggest environment set first.

\clearpage
\subsection{Benchmark Results} \label{sec:benchmark_results}
{\setlength{\tabcolsep}{6pt}
\begin{table}[H]
  \centering
  \begin{tabular}{lrrr|rrrrrrr}
    \toprule
    & \thead{$|I|$} & \thead{$|S|$}& \thead{$|A|$} & \thead{BFS-} & \thead{BFS+} & \thead{BFS-+} & \thead{DFS-} & \thead{DFS+} & \thead{DFS-+} & \thead{Comp.}\\
    \midrule
\multirow{5}{*}{\rotatebox{90}{Exponential}}
&8&19&9&0.1&0.1&0.1&0.1&0.1&0.1&0.1\\
&12&27&13&0.2&0.2&0.2&0.2&0.2&0.2&0.2\\
&16&35&17&1.5&2.9&1.6&1.7&2.6&2.4&2.7\\
&20&43&21&36.1&39.2&41.6&40.1&129.6&129.1&39.4\\
&24&51&25&MO&MO&MO&MO&TO&TO&MO\\

\cmidrule(lr){1-11}
\multirow{6}{*}{\rotatebox{90}{Frogger}}
&10&1200&4&0.2&0.2&0.2&0.2&0.2&0.2&20.0\\
&20&1200&4&0.4&0.4&0.4&0.4&0.4&0.4&MO\\
&29&1200&4&0.6&0.6&0.6&0.6&0.6&0.6&MO\\
&50&4000&4&2.7&2.7&2.7&2.6&2.6&2.6&MO\\
&80&4000&4&4.5&4.4&4.4&4.4&4.4&4.3&MO\\
&99&4000&4&6.1&6.0&6.2&6.0&5.9&5.9&MO\\

\cmidrule(lr){1-11}
\multirow{3}{*}{\rotatebox{90}{CMan}}
&256&625&4&5.9&7.0&5.8&6.8&6.7&6.6&7.0\\
&256&6561&4&39.5&43.2&39.9&40.6&40.6&40.0&46.2\\
&256&14641&4&85.5&86.5&86.0&84.1&82.7&82.6&97.1\\

\cmidrule(lr){1-11}
\multirow{3}{*}{\rotatebox{90}{PMan}}
&256&256&4&3.9&4.2&3.8&4.0&4.2&3.7&4.2\\
&256&1296&4&11.5&11.4&11.3&11.2&10.7&10.8&12.6\\
&256&4096&4&23.2&30.1&23.1&27.5&28.8&27.6&30.8\\

\cmidrule(lr){1-11}
\multirow{7}{*}{\rotatebox{90}{Grid}}
&9&122&4&0.1&0.1&0.1&0.1&0.1&0.1&0.1\\
&19&132&4&0.6&1.2&0.6&0.5&1.1&0.5&1.0\\
&39&152&4&6.4&MO&9.0&119.5&MO&145.1&MO\\
&79&192&4&MO&MO&MO&MO&MO&MO&MO\\
&99&374&4&MO&MO&MO&MO&MO&MO&MO\\
&149&424&4&MO&MO&MO&MO&MO&MO&MO\\
&199&474&4&MO&MO&MO&MO&MO&MO&MO\\

\cmidrule(lr){1-11}
\multirow{1}{*}{\rotatebox{90}{Inf}}
&5&4&2&0.0&0.0&0.0&0.0&0.0&0.0&0.0\\

\cmidrule(lr){1-11}
\multirow{4}{*}{\rotatebox{90}{MMind}}
&16&26&16&0.2&0.2&0.2&0.2&0.2&0.2&0.2\\
&27&21&27&0.5&0.9&0.5&0.4&0.7&0.3&0.8\\
&32&37&32&1.5&1.7&1.6&0.8&1.9&0.9&2.1\\
&81&31&81&50.1&166.6&49.5&8.2&TO&12.8&190.8\\
    \bottomrule
  \end{tabular}
  \caption{Satisfiable benchmark results (part 1)}
  \label{table:benchsat1}
\end{table}
}

{\setlength{\tabcolsep}{6pt}
\begin{table}[H]
  \centering
\scalebox{0.9}{
  \begin{tabular}{lrrr|rrrrrrrr}
    \toprule
    & \thead{$|I|$} & \thead{$|S|$}& \thead{$|A|$} & \thead{Ent-} & \thead{Ent+} & \thead{Ent-+} & \thead{NEnt-} & \thead{NEnt+} & \thead{NEnt-+} & \thead{SAT} & \thead{Belief}\\
    \midrule
\multirow{5}{*}{\rotatebox{90}{Exponential}}
&8&19&9&0.1&0.1&0.1&0.0&0.1&0.1&75.9&0.1\\
&12&27&13&0.2&0.2&0.2&0.2&0.2&0.2&TO&0.2\\
&16&35&17&1.8&2.5&2.4&2.7&7.3&7.5&TO&4.1\\
&20&43&21&40.1&132.0&131.4&34.5&196.3&197.6&TO&217.6\\
&24&51&25&MO&TO&TO&MO&MO&MO&TO&MO\\
\cmidrule(lr){1-12}
\multirow{6}{*}{\rotatebox{90}{Frogger}}
&10&1200&4&0.2&0.2&0.2&0.2&0.2&0.2&1.4&22.7\\
&20&1200&4&0.4&0.4&0.4&0.4&0.4&0.5&3.9&MO\\
&29&1200&4&0.6&0.6&0.6&0.6&0.6&0.7&9.1&MO\\
&50&4000&4&2.6&2.6&2.6&2.6&2.6&2.6&121.2&TO\\
&80&4000&4&4.3&4.4&4.4&4.3&4.3&4.4&597.3&TO\\
&99&4000&4&6.0&6.0&5.9&5.8&5.9&6.1&TO&TO\\

\cmidrule(lr){1-12}
\multirow{3}{*}{\rotatebox{90}{CMan}}
&256&625&4&6.6&6.6&6.6&5.8&7.2&5.9&TO&3.8\\
&256&6561&4&40.4&40.7&40.1&32.6&43.7&32.6&TO&9.1\\
&256&14641&4&84.0&82.4&82.5&65.4&88.8&65.3&TO&16.2\\

\cmidrule(lr){1-12}
\multirow{3}{*}{\rotatebox{90}{PMan}}
&256&256&4&4.0&4.2&3.9&3.5&4.2&3.8&11.2&4.6\\
&256&1296&4&10.8&10.9&10.9&9.1&11.4&9.1&375.7&19.5\\
&256&4096&4&27.2&28.7&27.3&21.0&30.9&21.1&TO&203.5\\

\cmidrule(lr){1-12}
\multirow{7}{*}{\rotatebox{90}{Grid}}
&9&122&4&0.1&0.1&0.0&0.1&0.1&0.1&0.2&0.2\\
&19&132&4&0.2&1.1&0.2&0.2&1.3&0.2&3.7&0.6\\
&39&152&4&0.4&MO&0.4&1.5&MO&1.6&121.3&12.6\\
&79&192&4&1.2&MO&1.1&MO&MO&MO&TO&MO\\
&99&374&4&4.9&MO&4.9&MO&MO&MO&TO&MO\\
&149&424&4&6.2&MO&6.2&MO&MO&MO&TO&MO\\
&199&474&4&6.3&MO&6.4&MO&MO&MO&TO&MO\\

\cmidrule(lr){1-12}
\multirow{1}{*}{\rotatebox{90}{Inf}}
&5&4&2&0.0&0.0&0.0&0.0&0.0&0.0&0.0&MO\\

\cmidrule(lr){1-12}
\multirow{4}{*}{\rotatebox{90}{MMind}}
&16&26&16&0.2&0.2&0.2&0.2&0.1&0.2&0.4&0.2\\
&27&21&27&0.6&0.7&0.7&0.6&0.7&0.7&TO&0.7\\
&32&37&32&1.6&1.9&1.6&1.5&1.6&1.5&TO&1.5\\
&81&31&81&153.8&TO&653.5&128.9&126.7&128.1&TO&250.4\\
    \bottomrule
  \end{tabular}
  }
  \caption{Satisfiable benchmark results (part 2)}
  \label{table:benchsat2}
\end{table}
}
\begin{table}[H]
  \centering
\scalebox{0.9}{
  \begin{tabular}{lrrr|rrrrrrr}
    \toprule
    & \thead{$|I|$} & \thead{$|S|$}& \thead{$|A|$} & \thead{BFS-} & \thead{BFS+} & \thead{BFS-+} & \thead{DFS-} & \thead{DFS+} & \thead{DFS-+} & \thead{Comp.}\\
    \midrule

\multirow{4}{*}{\rotatebox{90}{Frogger}}
&5&360&4&0.0&0.0&0.0&0.0&0.0&0.0&0.1\\
&10&360&4&0.1&0.1&0.1&0.1&0.1&0.1&4.2\\
&15&360&4&0.2&0.1&0.1&0.2&0.1&0.1&150.8\\
&29&360&4&0.4&0.3&0.3&0.3&0.3&0.3&MO\\

\cmidrule(lr){1-11}
\multirow{3}{*}{\rotatebox{90}{Pacman}}
&256&81&4&2.8&2.7&2.6&2.9&2.6&2.7&3.1\\
&256&625&4&5.6&5.3&5.3&5.7&5.4&5.5&7.2\\
&256&6561&4&30.4&30.4&30.4&30.8&30.5&30.4&46.2\\

\cmidrule(lr){1-11}
\multirow{4}{*}{\rotatebox{90}{MMind}}
&16&21&16&0.1&0.2&0.1&0.2&0.1&0.2&0.1\\
&27&17&27&0.5&0.5&0.5&0.5&0.5&0.5&0.6\\
&32&25&32&1.0&1.0&1.0&1.1&0.6&0.7&1.1\\
&81&21&81&37.4&37.3&37.4&43.2&56.2&56.0&76.8\\

\cmidrule(lr){1-11}
\multirow{6}{*}{\rotatebox{90}{Exponential}}
&12&26&13&0.2&0.2&0.2&0.2&0.1&0.1&0.2\\
&16&34&17&1.3&1.4&1.4&1.4&0.2&0.2&2.2\\
&20&42&21&32.3&34.9&37.9&33.8&0.3&0.3&33.3\\
&24&50&25&MO&MO&MO&MO&0.9&1.0&MO\\
&28&58&29&MO&MO&MO&MO&8.8&9.1&MO\\
&32&66&33&MO&MO&MO&MO&425.1&429.5&MO\\
\bottomrule
  \end{tabular}
}
  \caption{Unsatisfiable benchmark results (part 1)}
  \label{table:benchunsat1}
\end{table}
\begin{table}[H]
  \centering
\scalebox{0.9}{
  \begin{tabular}{lrrr|rrrrrrrr}
    \toprule
    & \thead{$|I|$} & \thead{$|S|$}& \thead{$|A|$} & \thead{Ent-} & \thead{Ent+} & \thead{Ent-+} & \thead{NEnt-} & \thead{NEnt+} & \thead{NEnt-+} & \thead{SAT} & \thead{Belief}\\

    \midrule

\multirow{4}{*}{\rotatebox{90}{Frogger}}
&5&360&4&0.0&0.0&0.0&0.0&0.0&0.0&TO&0.5\\
&10&360&4&0.1&0.1&0.1&0.1&0.1&0.1&TO&10.6\\
&15&360&4&0.1&0.1&0.2&0.2&0.2&0.2&TO&744.3\\
&29&360&4&0.3&0.3&0.3&0.3&0.3&0.3&TO&MO\\

\cmidrule(lr){1-12}
\multirow{3}{*}{\rotatebox{90}{Pacman}}
&256&81&4&2.8&2.6&2.6&2.7&2.6&2.5&TO&3.2\\
&256&625&4&5.5&5.3&5.2&5.6&5.2&5.0&TO&11.0\\
&256&6561&4&30.5&30.3&30.2&30.7&30.1&29.9&TO&TO\\

\cmidrule(lr){1-12}
\multirow{4}{*}{\rotatebox{90}{MMind}}
&16&21&16&0.1&0.1&0.1&0.1&0.1&0.2&TO&0.3\\
&27&17&27&0.5&0.5&0.5&0.5&0.5&0.5&TO&2.0\\
&32&25&32&1.1&0.7&0.6&0.9&0.9&0.9&TO&4.2\\
&81&21&81&41.9&47.1&41.1&39.8&38.7&38.6&TO&MO\\

\cmidrule(lr){1-12}
\multirow{6}{*}{\rotatebox{90}{Exponential}}
&12&26&13&0.2&0.1&0.1&0.2&0.2&0.2&TO&0.7\\
&16&34&17&1.4&0.2&0.2&1.3&1.3&1.4&TO&10.8\\
&20&42&21&33.8&0.8&0.8&31.0&163.5&173.8&TO&576.1\\
&24&50&25&MO&7.8&8.3&MO&MO&MO&TO&MO\\
&28&58&29&MO&329.7&334.3&MO&MO&MO&TO&TO\\
&32&66&33&MO&343.7&347.7&MO&MO&MO&TO&MO\\
    \bottomrule
  \end{tabular}
}
  \caption{Unsatisfiable benchmark results (part 2)}
  \label{table:benchunsat2}
\end{table}

\clearpage
\paragraph{Plots.}
\begin{figure}[H]
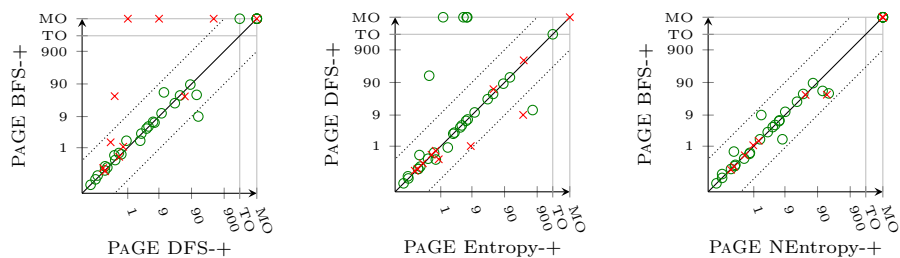

  \centering
  \makebox[\textwidth][c]{
  \subfigure {
    \benchmark{heuristicdfslowerupper}{\PaGE{} DFS-+}{heuristicbfslowerupper}{\PaGE{} BFS-+}
  }%
  \subfigure {
    \benchmark{heuristicentropylowerupper}{\PaGE{} Entropy-+}{heuristicdfslowerupper}{\PaGE{} DFS-+}
  }%
  \subfigure {
    \benchmark{heuristicnentropylowerupper}{\PaGE{} NEntropy-+}{heuristicbfslowerupper}{\PaGE{} BFS-+}
  }%
  }
  \vspace{-1em}
  \caption{Relative performance of \PaGE{} heuristics}
  \label{fig:extra_scatter_plots1}
\end{figure}
\begin{figure}[H]
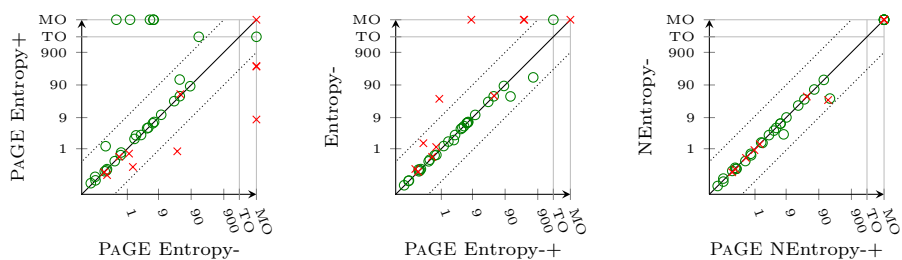

  \centering
  \makebox[\textwidth][c]{
  \subfigure {
    \benchmark{heuristicentropylower}{\PaGE{} Entropy-}{heuristicentropyupper}{\PaGE{} Entropy+}
  }%
  \subfigure {
    \benchmark{heuristicentropylowerupper}{\PaGE{} Entropy-+}{heuristicentropylower}{Entropy-}
  }%
  \subfigure {
    \benchmark{heuristicnentropylowerupper}{\PaGE{} NEntropy-+}{heuristicnentropylower}{NEntropy-}
  }%
  }
  \vspace{-1em}
  \caption{Relative performance of \PaGE{} algorithm using either a lower bound, upper bound, or both}
  \label{fig:extra_scatter_plots2}
\end{figure}

\end{document}